\documentclass[11pt,a4paper,headinclude,footinclude,chapterprefix=on]{scrreprt}

\usepackage{tabularx}
\usepackage{multicol}
\usepackage{graphicx}

\usepackage{tikz}
\usepackage{lipsum}

\usepackage[headsepline,plainheadsepline,footsepline,plainfootsepline]{scrpage2}

\usepackage[scaled=.92]{helvet}
\usepackage{courier}

\usepackage{subfig}
\usepackage{multirow}
\usepackage{algorithm}
\usepackage{algpseudocode}
\usepackage{newclude}

\usepackage{pdflscape}
\usepackage{afterpage}
\usepackage{color}

\usepackage{acronym}

\usepackage{footmisc}

\usepackage{upgreek}
\usepackage{amsmath,amssymb,amsthm,amsfonts,amsbsy}
\usepackage{mathtools}
\usepackage{bbm}
\usepackage{thmtools, thm-restate}

\usepackage[multiuser,marginclue,inline]{fixme}
\fxusetheme{color}
\fxsetup{targetlayout=colorcb}
\FXRegisterAuthor{awo}{anawo}{AWO}
\FXRegisterAuthor{mc}{anmc}{Mikolaj}
\FXRegisterAuthor{kmi}{ankmi}{}

\fxsetup{status=draft}

\usepackage{geometry}
\clearscrheadfoot

\ihead[\Large {\scshape TU Berlin }]{\Large {\scshape TU Berlin }}

\ifoot[{\tiny \begin{minipage}{5cm}Copyright at Technische Universit\"at Berlin.\newline All rights reserved.\end{minipage}}]{{\tiny \begin{minipage}{5cm}Copyright at Technische Universit\"at Berlin.\newline All rights reserved.\end{minipage}}}
\ofoot[Page \pagemark]{Page \pagemark}

\addtokomafont{pagefoot}{\normalfont}

\newcommand{\trnumber}{TKN-15-002}
\newcommand{\trdate}{March 2015}
\newcommand{\trauthor}{Niels Karowski, Konstantin Miller}
\newcommand{\tremail}{\{karowski,\,miller\}@tkn.tu-berlin.de}
\newcommand{\trtitle}{Optimized Asynchronous Passive Multi-Channel Discovery of Beacon-Enabled Networks}

\let\stdforall\forall
\renewcommand{\forall}{\mathop{\vphantom{\sum}\mathchoice
  {\vcenter{\hbox{\huge$\stdforall$}}}
  {\vcenter{\hbox{\Large$\stdforall$}}}{\stdforall}{\stdforall}}\displaylimits}

\graphicspath{{./}{figures/}{proofs/}{figures/eval/}{figures/numericEval/}}

\newcommand{\nie}[1]{\textcolor{red}{[#1]}}

\newcommand{\ALG}{{GREEDY}}

\newcommand{\GreedyRnd}{{GREEDY RND}}
\newcommand{\GreedyDeter}{{GREEDY DTR}}
\newcommand{\GreedyTrainRnd}{{GREEDY RND-SWT}}
\newcommand{\GreedyTrainDeter}{{GREEDY DTR-SWT}}
\newcommand{\GreedyTrain}{{GREEDY SWT}}
\newcommand{\OPTBTwo}{{$OPT_{B2}$}}

\newcommand{\evalFigWidth}{0.7}
\newcommand{\evalHspace}{-0.7em}

\begin{document}


{
\sffamily

\thispagestyle{empty}

\begin{tabularx}{\columnwidth}{cXc}
  \includegraphics[height=1cm]{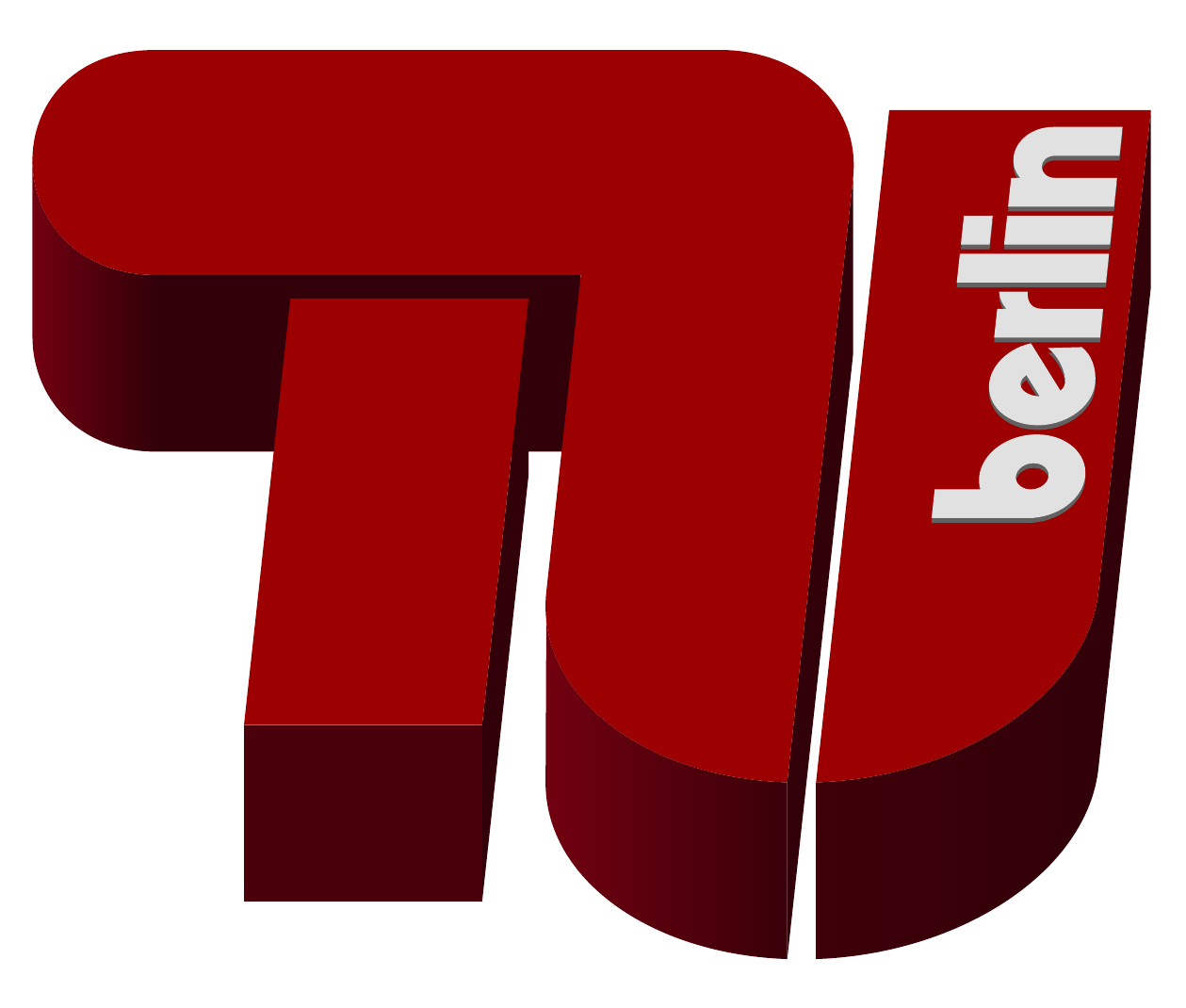}
  & &
  \includegraphics[height=1cm]{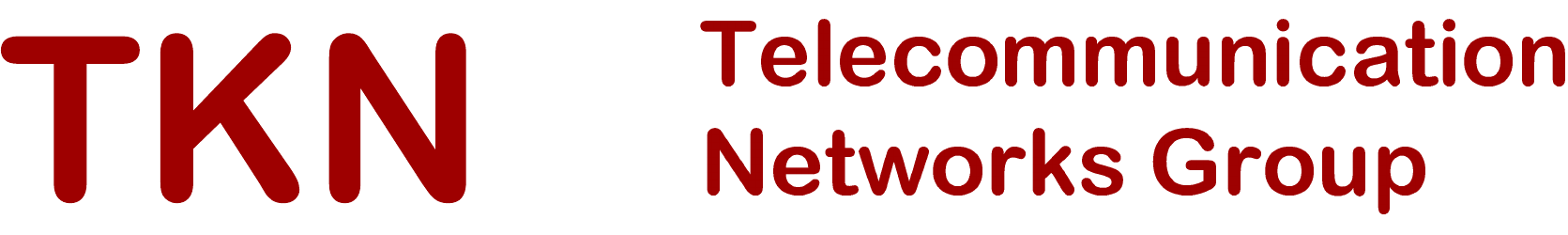}
  \\
\end{tabularx}

\vspace{1.0cm}

\begin{center}
{\huge
\noindent
Technische Universit\"at Berlin

\vspace{0.5cm}

\noindent
Telecommunication Networks Group

\begin{center}
\rule{15.5cm}{0.4pt}
\end{center}
}
\end{center}

\begin{minipage}[][11.0cm][c]{14.5cm}
{\Huge

\begin{center}
\trtitle
\end{center}

\begin{center}
{\LARGE \trauthor} \\
{\Large \tremail}
\end{center}

\begin{center}
Berlin, \trdate
\end{center}

\vspace{0.5cm}

}

\begin{center}
\setlength{\fboxrule}{2pt}\setlength{\fboxsep}{2mm}
\fbox{TKN Technical Report \trnumber}
\end{center}

\end{minipage}

\setlength{\fboxrule}{0.4pt}
\setlength{\fboxsep}{0.4pt}

\begin{center}

  \rule{15.5cm}{0.4pt}

  \vspace{0.5cm}

  {\huge {TKN Technical Reports Series}}

  \vspace{0.5cm}

  {\huge Editor: Prof. Dr.-Ing. Adam Wolisz}

  \vspace{0.5cm}

 \end{center}

}

\begin{abstract}
\subsection*{\abstractname}

Neighbor discovery is a fundamental task for wireless networks deployment. It is essential for setup and maintenance of networks and is typically a precondition for further communication. In this work we focus on passive discovery of networks operating in multi-channel environments, performed by listening for periodically transmitted beaconing messages. It is well-known that performance of such discovery approaches strongly depends on the structure of the adopted \acf{BI} set, that is, set of intervals between individual beaconing messages. However, although imposing constraints on this set has the potential to make the discovery process more efficient, there is demand for high-performance discovery strategies for \ac{BI} sets that are as general as possible. They would allow to cover a broad range of wireless technologies and deployment scenarios, and enable network operators to select \acp{BI} that are best suited for the targeted application and/or device characteristics.

In the present work, we introduce a family of novel low-complexity discovery algorithms that minimize both the \acf{EMDT} and the makespan, for a quite general family of \ac{BI} sets. Notably, this family of \ac{BI} sets completely includes \acp{BI} supported by IEEE~802.15.4 and a large part of \acp{BI} supported by IEEE~802.11. Furthermore, we present another novel discovery algorithm, based on an \ac{ILP}, that minimizes \ac{EMDT} for arbitrary \ac{BI} sets.


In addition to analytically proving optimality, we numerically evaluate the proposed algorithms using different families of \ac{BI} sets and compare their performance with the passive scan of IEEE~802.15.4 w.r.t. various performance metrics, such as makespan, \ac{EMDT}, energy usage, etc.



\end{abstract}

\tableofcontents


\chapter{Introduction}

We are currently at the brink of a new era of computing driven
by rapid augmentation of physical objects around us with
computational and wireless communication capabilities. The
resulting network of “smart” objects that interact and exchange
information without direct human intervention, the so-called
Internet of Things (IoT), can offer deep real-time awareness and
control of the physical environment and serve as foundation for
novel applications in a wide range of domains. 
It is estimated that 25~\cite{gartnerIoT} to 50~\cite{cisco_IoT} billion devices and objects will be connected to the Internet by 2020. Since the vast majority of these devices will communicate wirelessly, this development raise the demand for efficient neighbor discovery approaches. Discovering neighbor networks in a wireless network environment is typically the first step of communication initialization. In addition, if performed periodically, it helps avoiding conflicts by, e.g., selecting a less occupied channel. Note that since most devices will be battery powered, neighbor discovery must be performed in an energy-efficient way.

This paper focuses on the problem of discovering neighbors in a fast and energy-efficient manner by passively listening to periodically transmitted beaconing messages of neighbors that are agnostic to the discovery process. Neighbors may operate on 
different channels using various \acp{BI} based on their targeted application, used hardware, and operational state (e.g., dynamically adapting \acp{BI} to current battery level). Since devices are typically not synchronized and have no a priori knowledge about their environment, designing schedules determining when to listen, on which channel, and for how long, that allow for a fast discovery of all neighbors, is a challenging task.

State-of-the-art beacon-enabled wireless network technologies include, e.g., IEEE~802.11~\cite{ieee80211} and IEEE~802.15.4~\cite{ieee802154}. With these technologies, beacon messages are transmitted periodically for management and time synchronization purposes. The idea of the presented discovery approaches is to use periodic transmission of beacon frames for the discovery process. In contrast, discovery approaches requiring control over the transmission of discovery messages might be in conflict with the deployed \ac{MAC} protocol, while reusing beacon messages allows for more flexibility in working with a broad range of technologies.

It is well-known that performance of discovery approaches strongly depends on the structure of the \ac{BI} set. Although restricting the latter facilitates the discovery process, it is necessary to develop discovery strategies that are as nonrestrictive as possible, in order to cover a broad range of existing wireless technologies, and to allow network operators to select \acp{BI} that are best suited for the targeted application and/or device characteristics. 

Our previous work~\cite{Karowski11, Karowski13} focused on discovery of IEEE~802.15.4 networks, where \acp{BI} are restricted to the form $\tau\cdot 2^{BO}$, where $\tau$ is the duration of a superframe, and \acf{BO} is a parameter taking values between 0 and 14. 
In contrast, in the present work, we focus on developing efficient algorithms for broader families of \ac{BI} sets, including arbitrary \ac{BI} sets that do not have any restrictions.

In particular, we develop a family of novel low-complexity algorithms that "greedily" maximize the number of discovered networks independently in each subsequent time slot. We identify an important family of \ac{BI} sets, where each \ac{BI} is an integer multiple of all smaller \acp{BI}, and prove that for \ac{BI} sets from this family, greedy algorithms are optimal w.r.t. the \ac{EMDT} and the makespan. Notably, the identified family completely includes \ac{BI} sets supported by IEEE~802.15.4 and a large part of \ac{BI} sets supported by IEEE~802.11.

In addition, we present a novel discovery approach, based on an \ac{ILP}, that minimizes \ac{EMDT} for arbitrary \ac{BI} sets. Although quite attractive due to the broad range of supported \acp{BI}, this approach has high computational complexity and memory consumption, such that its usage is restricted to offline computations, and to network environments with a limited number of channels and small \acp{BI}.

In addition to analytically proving optimality of the proposed discovery algorithms, we numerically evaluated them using different families of \ac{BI} sets, and compare their performance with the passive scan, specified as part of the IEEE~802.15.4 standard, w.r.t. various performance metrics, such as makespan, \ac{EMDT}, energy usage, etc.


Finally, we discuss the strong impact the structure of permitted \ac{BI} sets has on the performance of discovery approaches, and provide recommendations for \ac{BI} selection that supports efficient neighbor discovery. This recommendation might be useful, on the one hand, for the development of novel wireless communication based technologies using periodic beacon frames for management or synchronization purposes, and, on the other hand, for the deployment of existing technologies that support a wide range of \acp{BI}, such as, e.g., IEEE~802.11. 

In summary, our main contributions are as follows.

\begin{itemize}
\item A family of low-complexity discovery algorithms that minimize \ac{EMDT} \textbf{and} makespan for the broad family of \ac{BI} sets where each \ac{BI} is an integer multiple of all smaller \acp{BI}. Notably, this family completely includes \acp{BI} supported by the standard IEEE~802.15.4 and a large part of \ac{BI} sets supported by IEEE~802.11.
\item An \ac{ILP}-based approach to computing listening schedules minimizing \ac{EMDT} for arbitrary sets of \acp{BI}.
\item A recommendation on the selection of \ac{BI} sets to support the discovery process.
\end{itemize}

The rest of this paper is structured as follows. 
Section~\ref{sec:system} presents our system model, accompanied by basic definitions and preliminary results. 
In Section~\ref{sec:BIFamilies} we present the families of \ac{BI} sets studied in the present work.
Section~\ref{sec:preliminaries} introduces the notion of \ac{EMDT} and establishes an upper bound on the time required to minimize it.
Section~\ref{sec:strategies} introduces the developed discovery strategies, accompanied by analytical results on their optimality. 
Section~\ref{sec:numerical_eval} presents evaluation settings and results.
In Section~\ref{sec:biSelection}, we provide recommendations on the selection of \acp{BI} supporting the discovery process. 
Finally, Section~\ref{sec:conclusion} concludes this paper and outlines future work.

\pagebreak

\chapter{System Model and Basic Definitions}
\label{sec:system}

We consider a set of networks $N$, operating on a set of channels $C$. Each network $\nu\in N$ is assigned a fixed channel $c_\nu\in C$. Each network announces its presence using beacon signaling messages (beacons, in the following) that, e.g., are sent by a selected device acting as network coordinator. We assume that beacons are sent periodically with a fixed interval of $b_\nu\cdot\tau$ seconds, $b_\nu\in B$, where $B\subset\mathbb{N}^+$ is a finite set of beacon intervals, and $\tau$ is a technology dependent parameter. The set of \acp{BI} $B$ might be determined by, e.g., the choice of communication technology or by specifications provided by network operators. A characterization of \ac{BI} sets studied in the present work will be provided in Chapter~\ref{sec:BIFamilies}.

We divide time into slots of length $\tau$, and number them starting with 1, such that $i$-th time slot contains the time period $\left[(i-1)\tau,\,i\tau\right]$. The parameter $\tau$ is typically determined by the choice of technology. For IEEE~802.15.4, $\tau$ might be the duration of a base superframe, which is approximately 15.36 $ms$ (960 symbols) when operating in the 2.4~GHz frequency band. In case of IEEE~802.11, $\tau$ might be the duration of a \ac{TU}, which equals 1024 $\mu s$. We assume that the maximum beacon transmission time (time required to send one beacon) is smaller than $\tau$. We denote the offset, that is, the first time slot, in which network $\nu$ sends a beacon, by $\delta_\nu\leq b_\nu$. We denote the set of beaconing time slots of a network $\nu$ by $\mathcal{T}_\nu=\left\{\delta_\nu+i \cdot b_\nu\right\}_{i\geq 0}$. We call the resulting set of (channel, time slot) pairs, $\mathcal{B}_\nu=\left\{c_\nu\right\}\times\mathcal{T}_\nu$, the beacon schedule of network $\nu$.

With the given notation, each network $\nu$ is assigned a tuple $\left(c_\nu,b_\nu,\delta_\nu\right)$, which we call a network configuration. Note that assigning multiple networks the same network configuration does not necessarily lead to beacon collisions due to the assumption that the maximum beacon transmission time is shorter than $\tau$. With IEEE~802.15.4, e.g., default beacon size when operating in the 2.4 GHz frequency band is 38 symbols, as compared to the time slot duration $\tau = 960$ symbols.

We denote the set of possible network configurations for a given \ac{BI} set $B$ and a given set of channels $C$ by $K_{BC}=\left\{\left(c,b,\delta\right)\;|\;c\in C,\,b\in B,\,\delta\in\left\{1,\ldots,b\right\}\right\}$. We define a network environment to be a function $E:N\mapsto K_{BC}$ assigning each network a network configuration. Slightly abusing notation, in the following, we use index $\nu$ to refer to a particular network, and index $\kappa$ to refer to a particular network configuration. Thus, we denote a network configuration $\kappa\in K_{BC}$ by $\kappa=\left(c_\kappa,b_\kappa,\delta_\kappa\right)$.
 Analogously, we define $\mathcal{T}_\kappa$ and $\mathcal{B}_\kappa$ as the set of beaconing time slots and the beacon schedule of any network using network configuration $\kappa$.
 
For a time slot $t$ and a \ac{BI} $b$, $\delta_b(t)=\left(t\,\text{mod}\,b\right)+1$ shall denote the unique offset such that a network with configuration $\left(c,b,\delta_b(t)\right)$, $c\in C$, transmits its beacon in time slot $t$. Observe that $\delta_b(t)$ has periodicity $b$, that is, $\delta_b(t)=\delta_b(t+b)$ for each $t$, $b$. Consequently, vector function $\delta(t)=\left(\delta_b(t),\,b\in B\right)$ has periodicity $LCM(B)$, that is, $\delta(t)=\delta\left(t+LCM(B)\right)$ for each $t$. 


We assume that initially each device knows its own network configuration, the set of available \acp{BI} $B$ and the set of channels $C$. In addition, it might know probabilities $\left(p_\kappa,\,\kappa\in K_{BC}\right)$ that a neighbor network might use configuration $\kappa$. Equipped with this knowledge, in order to detect neighbor networks, a device performs neighbor discovery by selecting time slots during which it listens on particular channels in order to overhear beacons transmitted by neighbors, starting with time slot 1. We call the resulting set of (channel, time slot) pairs a listening schedule, denoted by $\mathcal{L}\subset C\times\mathbb{N}$. Since we assume that devices cannot simultaneously listen on multiple channels, we demand $c\neq c'\Rightarrow t\neq t'$ for all $\left(c,\,t\right),\left(c',\,t'\right)\in \mathcal{L}$. 

In a typical network environment, not all configurations will be used, while, at the same time, some configurations might be used by more than one network. We call an environment complete, if each configuration is used by exactly one network.

\begin{restatable}[Complete environment]{definition}{completeenvironment}
\label{def:complete_environment}
For a \ac{BI} set $B\subset\mathbb{N}^+$ and a set of channels $C$, an environment $E:N\mapsto K_{BC}$ is called complete if and only if $E$ is bijective (for each $\kappa\in K_{BC}$ there exists exactly one $\nu\in N$ with $\kappa=\left(c_\nu,b_\nu,\delta_\nu\right)$).
\end{restatable}

Under ideal conditions, when beacons are never lost and devices are not mobile, if a schedule $\mathcal{L}$ contains at least one element from the beacon schedule $\mathcal{B}_\kappa$ for each configuration $\kappa\in K_{BC}$, it allows to discover all neighbor networks in a complete environment. We call such a schedule \emph{complete}.

\begin{restatable}[Complete schedule]{definition}{completeschedule}
\label{def:complete_schedule}
For a \ac{BI} set $B\subset\mathbb{N}^+$ and a set of channels $C$, a schedule $\mathcal{L}\subset C\times\mathbb{N}$ is called complete if and only if $\mathcal{L}\cap\mathcal{B}_\kappa\neq\emptyset$, $\forall \kappa\in K_{BC}$.
\end{restatable}

For a complete schedule $\mathcal{L}$, we denote by $T_\nu\left(\mathcal{L}\right)=\min\left\{t\in\mathcal{T}_\nu\;|\;\left(c_\nu,\,t\right)\in\mathcal{L}\right\}$ the discovery time of network $\nu$. Similarly, we denote by $T_\kappa\left(\mathcal{L}\right)=\min\left\{t\in\mathcal{T}_\kappa\;|\;\left(c_\kappa,\,t\right)\in\mathcal{L}\right\}$ the discovery time of all networks operating with configuration $\kappa$. Whenever the considered schedule is clear from the context, we might simply write $T_\nu$ or $T_\kappa$. An important performance metric for a listening schedule is the time it requires to detect all neighbor networks. We call this time the makespan of a schedule.
\begin{restatable}[Makespan of a schedule]{definition}{makespan}
\label{def:makespan}
For a \ac{BI} set $B\subset\mathbb{N}^+$, and a set of channels $C$, we call the time slot of a complete schedule $\mathcal{L}$ during which the last configuration is detected the makespan of $\mathcal{L}$ and denote it by $T_\mathcal{L}$. That is, $T_\mathcal{L}=\max_{\kappa\in K_{BC}}T_\kappa\left(\mathcal{L}\right)$. 
\end{restatable}

The following proposition and corollary establish a lower bound on the makespan.
\begin{restatable}{proposition}{premaxBC}
\label{prop:premaxBC}
For a \ac{BI} $b\in B\subset\mathbb{N}^+$, an offset $\delta\in\{1,\ldots,b\}$, a set of channels $C$, and a complete schedule $\mathcal{L}$, the earliest time slot until which all configurations $\left\{\left(c,b,\delta\right)\,|\,c\in C\right\}$ can be discovered is $b\left(\left|C\right|-1\right)+\delta$. 
\end{restatable}
\begin{proof}
Note that $\lvert\left\{\left(c,b,\delta\right)\,|\,c\in C\right\}\rvert=\left|C\right|$. Since each configuration in this set has its beacons on a different channel, the number of time slots that have to be scanned cannot be smaller than $\left|C\right|$. Observe that the earliest time slot when this number of corresponding time slots can be reached is $b\left(\left|C\right|-1\right)+\delta$, proving the claim.
\end{proof}

\begin{restatable}{corollary}{maxBC}
\label{cor:maxBC}
For an arbitrary set of \acp{BI} $B\subset\mathbb{N}^+$, a set of channels $C$, and a complete schedule $\mathcal{L}$, it holds $T_{\mathcal{L}}\geq \max(B)\cdot\lvert C\rvert$. We call complete schedules with makespan $\max(B)\cdot\lvert C\rvert$ makespan-optimal.
\end{restatable}
\begin{proof}
The claim is a direct consequence of Proposition~\ref{prop:premaxBC}, with $b=\delta=\max(B)$.
\end{proof}

In addition to minimizing makespan, it is often desirable to minimize the \ac{MDT} of a schedule, defined in the following.
\begin{restatable}[\ac{MDT} of a schedule]{definition}{mdtdef}
\label{def:mdt}
For a set of networks $N$, a \ac{BI} set $B\subset\mathbb{N}^+$, a set of channels $C$, and a network environment $E$, the \ac{MDT} of a complete schedule $\mathcal{L}$ is given by $$\frac{1}{\lvert N\rvert} \sum_{\nu\in N}T_\nu\left(\mathcal{L}\right)\,.$$ For a complete environment, \ac{MDT} can also be computed as 
$$\frac{1}{\lvert N\rvert} \sum_{\nu\in N}T_\nu\left(\mathcal{L}\right)=\frac{1}{\lvert K_{BC}\rvert}\sum_{\kappa\in K_{BC}}T_\kappa\left(\mathcal{L}\right)=\frac{1}{\lvert C\rvert\lvert B\rvert\sum_{b\in B}b}\sum_{\kappa\in K_{BC}}T_\kappa\left(\mathcal{L}\right)\,.$$
\end{restatable}

Typically, however, a device does not know its network environment (otherwise, it would not have to perform a discovery) so that \ac{MDT} cannot be computed and optimized a priori. Nevertheless, a device still can minimize the expected value of the \ac{MDT} given probabilities that a neighbor network is using configuration $\kappa\in K_{BC}$. We will address this question in Chapter~\ref{sec:preliminaries}.

For complete environments, we define a special type of schedules that we call recursive and that we will use to prove optimality of the proposed discovery strategies in the following chapters.
\begin{restatable}[Recursive schedule]{definition}{recursiveschedule}
\label{def:recursive_schedule}
In a complete environment with \ac{BI} set $B\subset\mathbb{N}^+$, and a set of channels $C$, a schedule $\mathcal{L}\subset C\times\mathbb{N}$ is called recursive if and only if 
\begin{itemize}
\itemsep0em
\item for any $t\in\left[1,\,\max(B)\cdot\lvert C\rvert\right]$ there exists a $c\in C$ such that $(c,t)\in\mathcal{L}$ (no idle slots)
\item for any $b\in B$, and for any $t,t'\in\left[1,\,b\cdot\lvert C\rvert\right]$ with $\delta_b(t)=\delta_b(t')$, if $(c,t),\,(c',t')\in\mathcal{L}$ then $c\neq c'$ (no redundant scans)
\end{itemize}
\end{restatable}

\begin{figure}
\begin{center}
\centerline{
  \includegraphics[width=0.5\textwidth]{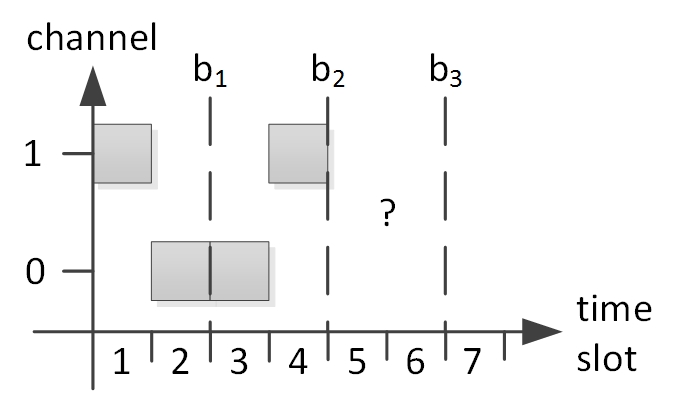}
}
\caption{Example of the non-existence of a structured schedule showing the beginning of the listening schedule depicted by gray boxes for $B = \{1,2,3\}$ and $|C| = 2$.}
 \label{fig:NonStructuredSchedule}
\end{center}
\end{figure}

The following proposition provides a characterizing property of recursive schedules.

\begin{restatable}{proposition}{recursivecharacterization}
\label{prop:recursive_characterization}
In a complete environment with a \ac{BI} set $B=\left\{b_1,\ldots,b_n\right\}\subset\mathbb{N}^+$, $b_i<b_j$ for $i<j$, and a set of channels $C$, a schedule $\mathcal{L}\subset C\times\mathbb{N}$ is called recursive if and only if each scan $(c,t)\in\mathcal{L}$, with $t\in\left[1,\,b\cdot\lvert C\rvert\right]$, results in the discovery of network configuration $\left(c,b,\delta_b(t)\right)$. Alternatively, a schedule is recursive if and only if in each time slot $t\in\left[1,\,b_i\cdot\lvert C\rvert\right]$, for a $b_i\in B$, it discovers at least $n-i+1$ configurations.
\end{restatable}
\begin{proof}
Follows directly from the definition of a recursive schedule.
\end{proof}

Recursive schedules have a compelling property that they are always complete, makespan-optimal, and \ac{MDT}-optimal, as stated in the following Corollary. In the following chapters, we will define families of \ac{BI} sets where recursive schedules always exist, and present algorithms for their efficient computation.

\begin{restatable}{corollary}{recursiveoptimal}
\label{cor:recursive_optimal}
In a complete environment with a \ac{BI} set $B\subset\mathbb{N}^+$, a recursive schedule is complete, makespan-optimal, \ac{MDT}-optimal, and maximizes the number of configurations detected until each time slot, for each $B'\subset B$.
\end{restatable}
\begin{proof}
Follows from Proposition~\ref{prop:recursive_characterization}.
\end{proof}

It is worth noting that a recursive schedule does not always exist. Consider the example illustrated in Figure~\ref{fig:NonStructuredSchedule}, with $B = \{1,2,3\}$ and $|C| = 2$. Vertical dashed lines indicate the time slots until all configurations $\left(c,b_i,\delta\right)$ for a \ac{BI} $b_i$ have to be completely discovery on each channel $c \in C$ in order for the schedule to be recursive. Consequently, gray boxes indicate channels that have to be scanned during the first 4 time slots (uniquely determined up to swapping channel 0 and 1). Observe that a recursive listening schedule for this example does not exist due to the fact that it would have to discover all remaining network configurations using \ac{BI} $b_3$ during only two time slots 5 and 6. However, there are three remaining configurations: $(0, 3, 1)$, $(1, 3, 2)$, and $(1, 3, 3)$.

Finally, we would like to remark that the analytical optimality results in the following are obtained under three idealizing assumptions. 
\begin{itemize}
\itemsep0em 
	\item Switching between channels is performed instantaneously (switching time is 0).
	\item Beacon transmission/reception time is 0.
	\item There are no beacon losses.
\end{itemize}
It is ongoing work to evaluate the performance of the proposed algorithms by means of simulations and experiments in real network environments.

Table~\ref{tab:system_model_parameter} provides a summary of all defined parameters.

\begin{table*}
\caption{Summary of the defined parameters}
\label{tab:system_model_parameter}
\centering
	\begin{tabular}{p{5cm}p{7cm}}
	\hline
			Name & Description \\
			\hline
			\textbf{General}																													& \\
			$\tau$																													& Time slot duration\\
			$C\subset\mathbb{N}^+$																					& Set of channels \\
			$B\subset\mathbb{N}^+$																					& Set of \acp{BI}\\
			$LCM(B)$																												& Least common multiple of a set $B$\\	
			& \\
			\textbf{Network} &\\
			$N$																															& Set of networks \\
			$\nu \in N$																											& Network $\nu$ \\
			$c_\nu \in C$																										& Fixed operating channel of network $\nu$\\
			$b_\nu \in B$																										& \ac{BI} of network $\nu$ \\
			$\delta_\nu \in \{1, \ldots, b_\nu \}$																					& Beacon offset of network $\nu$\\
			$\left(c_\nu,b_\nu,\delta_\nu\right)$														& Configuration of network $\nu$  \\
			$\mathcal{T}_\nu=\left\{\delta_\nu+i \cdot b_\nu\right\}_{i\geq 0}$	  & Set of beaconing time slots of network $\nu$\\
			$\mathcal{B}_\nu=\left\{c_\nu\right\}\times\mathcal{T}_\nu$		  & Beacon schedule of network $\nu$ \\
			$\begin{aligned} T_\nu\left(\mathcal{L}\right)= & \min\left\{t\in\mathcal{T}_\nu\;|  \right. \\ &\left.\left(c_\nu,\,t\right)\in\mathcal{L}\right\} \end{aligned}$ & Discovery time of network $\nu$, given listening schedule $\mathcal{L}$ \\	
			& \\
			\textbf{Network configuration} & \\
			$\begin{aligned} K_{BC}= & \left\{\left(c,b,\delta\right)\;|\;c\in C,\,  \right. \\ &\left. b\in B,\,\delta\in\left\{1,\ldots,b\right\}\right\} \end{aligned}$ & Set of possible network configurations for a given \ac{BI} set $B$ and a given set of channels $C$\\
			$\kappa=\left(c_\kappa,b_\kappa,\delta_\kappa\right) \in K_{BC}$	& Network configuration $\kappa$ using \ac{BI} $b_\kappa$, channel $c_\kappa$ and offset $\delta_\kappa$ \\
			$\mathcal{T}_\kappa=\left\{\delta_\kappa+i \cdot b_\kappa\right\}_{i\geq 0}$	  & Set of beaconing time slots of networks operating with configuration $\kappa$\\
			$\mathcal{B}_\kappa=\left\{c_\kappa\right\}\times\mathcal{T}_\kappa$		  			& Beacon schedule of networks operating with configuration $\kappa$ \\
			$\begin{aligned} T_\kappa\left(\mathcal{L}\right)= & \min\left\{t\in\mathcal{T}_\kappa\;| \right. \\ &\left.\left(c_\kappa,\,t\right)\in\mathcal{L}\right\} \end{aligned}$ & Discovery time of all networks operating with configuration $\kappa$ \\
			& \\
			\textbf{Listening schedule} & \\
			$\mathcal{L}\subset C\times\mathbb{N}$													& Listening schedule consisting of a sequence of (channel, time slot) pairs\\
			$T_\mathcal{L}=\max_{\kappa\in K_{BC}}T_\kappa\left(\mathcal{L}\right)$ & Makespan of a listening schedule $\mathcal{L}$\\
	\hline
	\end{tabular}
\end{table*}

\chapter{Families of Beacon Interval Sets}
\label{sec:BIFamilies}

In this section we define several families of \ac{BI} sets, studied in the present work. They allow us to formulate properties of developed discovery strategies as functions of \ac{BI} sets. The defined families are characterized in the following table and illustrated in Figure~\ref{fig:BIFamilies}.

\def\tabularxcolumn#1{m{#1}}
\begin{table}[h]
\begin{center}
\begin{tabularx}{\textwidth}{|| l | X ||}
\hline\hline
$\mathbb{F}_1$ & This is the most general family of \ac{BI} sets that contains any finite subset of $\mathbb{N}^+$. \\\hline
$\mathbb{F}_2$ & Family $\mathbb{F}_2\subset\mathbb{F}_1$ includes all \ac{BI} sets, for which the maximum element equals the \ac{LCM} of the set: $\max(B)=LCM(B)$, that is, the maximum beacon interval is an integer multiple of all other beacon intervals in the set. \\\hline
$\mathbb{F}_3$ & For any \ac{BI} set $B\in\mathbb{F}_3\subset\mathbb{F}_2$, we demand that $\max(B')=LCM(B')$ holds for any subset $B'\subseteq B$. An equivalent formulation is that for any $b,b'\in B\in\mathbb{F}_3$ with $b<b'$ there exists an $\alpha\in\mathbb{N}^+$ such that $b'=\alpha\cdot b$. That is, any beacon interval is an integer multiple of any smaller beacon interval in the set. \\\hline
$\mathbb{F}_4$ & Family $\mathbb{F}_4\subset\mathbb{F}_3$ includes all sets whose elements are powers of the same base, potentially multiplied by an common coefficient. That is, for a $B=\left\{b_1,\ldots,b_n\right\}\in\mathbb{F}_4$ we demand that there exist $k,c\in\mathbb{N}^+$ and $e_1,\ldots,e_n\in\mathbb{N}$ such that $b_i=kc^{e_i}$, $\forall i\in\{1,\ldots,n\}$. \\\hline
$\mathbb{F}_{\textrm{IEEE 802.15.4}}$ & This family contains all sets defined by the IEEE~802.15.4 standard. All elements $b\in B\in\mathbb{F}_{\textrm{IEEE 802.15.4}}$ must have the form $b=2^{BO}$, where $BO$ is a network parameter called the beacon order that can be assigned a value between 0 and 14. Possible intervals lie in the range between approx. 15.36 $ms$ and 252 $s$. \\\hline
$\mathbb{F}_{\textrm{IEEE 802.11}}$ & The family of \ac{BI} sets allowed by the IEEE~802.11 standard contains arbitrary sets such that each element $b\in B\in \mathbb{F}_{\textrm{IEEE 802.11}}$ can be represented by a 16 bit field, that is, $b\in\left[1,\,2^{16}-1\right]$. \\
\hline\hline
\end{tabularx}
\end{center}
\label{tab:sets}
\end{table}

\begin{figure}[t]
\captionsetup{justification=centering}
\begin{center}
\centerline{
  \includegraphics[width=0.7\textwidth]{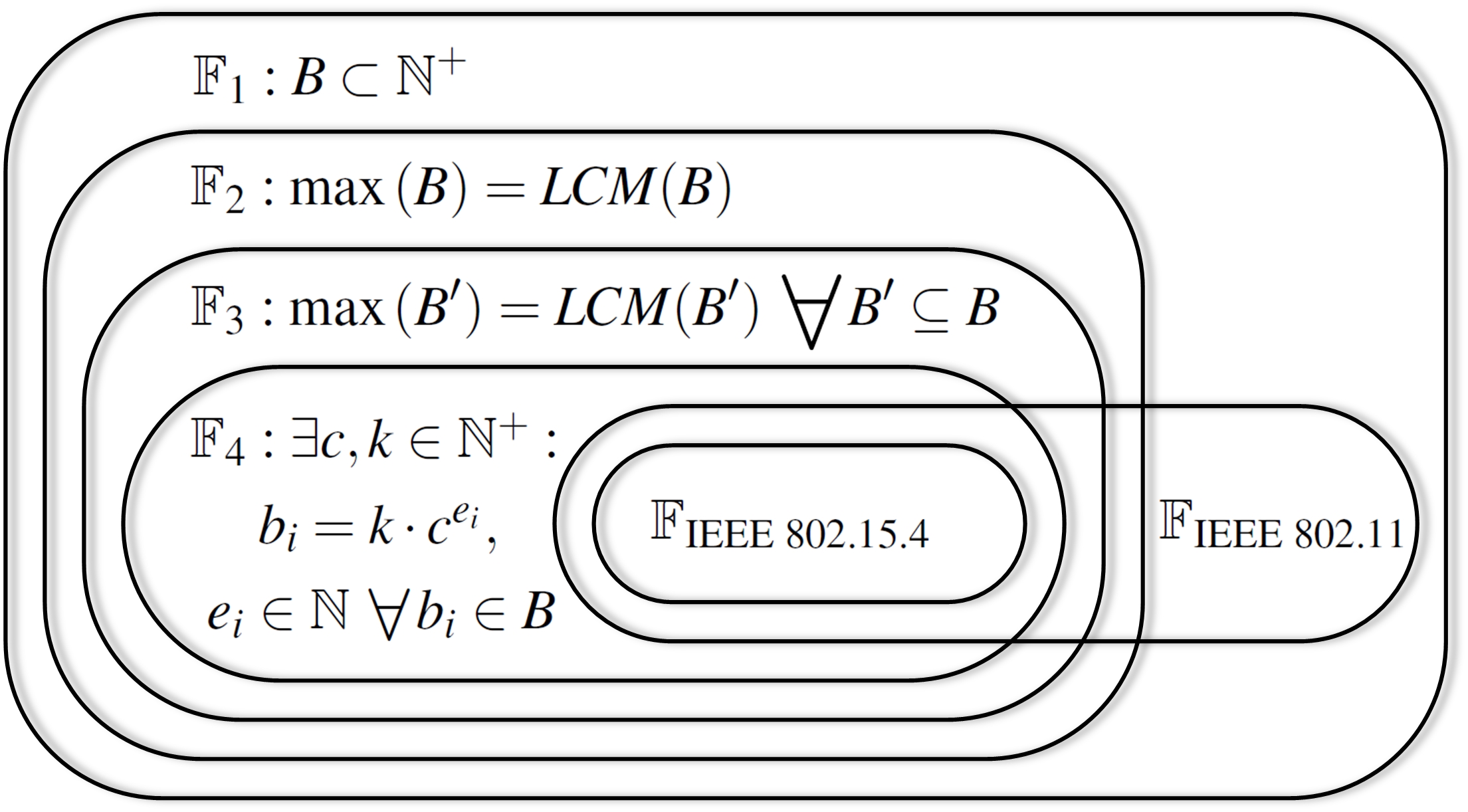}
}
\caption{Studied families of \ac{BI} sets.}
 \label{fig:BIFamilies}
\end{center}
\end{figure}

An important property of the family of \ac{BI} sets $\mathbb{F}_3$ is that in complete network environments using \ac{BI} sets from $\mathbb{F}_3$ there always exists a recursive schedule, as shown by the following proposition.

\begin{restatable}{proposition}{recursiveF3}
\label{prop:recursive_F3}
In a complete environment with a \ac{BI} set $B \in \mathbb{F}_3$, and a set of channels $C$, a recursive schedule always exists.
\end{restatable}
\begin{proof}
To prove the claim we will use the characteristic property of recursive schedules from Proposition~\ref{prop:recursive_characterization}. Assume w.l.o.g. $B=\left\{b_1,\ldots,b_n\right\}$ with $b_i<b_j$ for $i<j$. Observe that it is possible to discover one configuration $\left(c,b_1,\delta_{b_1}(t)\right)$ in each time slot $t\in\left[1,\,b_1\cdot\lvert C\rvert\right]$, since for each offset $\delta\in\left\{1,\ldots,b_1\right\}$ there are $\lvert C\rvert$ time slots in the interval $\left[1,\,b_1\cdot\lvert C\rvert\right]$ where this offset occurs, and since the sets of time slots for individual offsets to not intersect.

Further, observe that in $\mathbb{F}_3$, each discovery of a configuration $\left(c,b_i,\delta_{b_i}(t)\right)$ results in a discovery of $\left(c,b_j,\delta_{b_j}(t)\right)$ for each $j\in\{i+1,\ldots,n\}$. Consequently, by induction, for each $b_i\in B$, one configuration $\left(c,b_i,\delta\right)$ can be discovered in each time slot $t\in\left[1,\,b_{i-1}\cdot\lvert C\rvert\right]$, while, analogously to the argumentation above, one of the remaining configurations can be discovered in each of the time slots $t\in\left[b_{i-1}\cdot\lvert C\rvert+1,\,b_i\cdot\lvert C\rvert\right]$.
\end{proof}

In the following, w.l.o.g., we only consider \ac{BI} sets whose \ac{GCD} is 1. Indeed, please observe that a listening schedule for a \ac{BI} set with \ac{GCD} $d\neq 1$ is equivalent to a listening schedule for the transformed \ac{BI} set where each element is divided by $d$, and $\tau$ is substituted by $\tau'=\tau\cdot d$. This transformation allows to reduce the computational complexity, especially for \ac{ILP}-based approaches.

\chapter{Preliminaries}
\label{sec:preliminaries}

A desired goal for a discovery strategy is to minimize the \acf{MDT}, where mean is taken over all networks in a particular environment. Since, however, a device does not a priori know which network configurations are present in its environment, this problem is not solvable. Nevertheless, it is possible to minimize the expected value for \ac{MDT}, called \acf{EMDT} in the following, given the assumption that certain network configurations are picked by individual networks with certain probabilities. In the following, we assume uniform probabilities for channels, \acp{BI}, and \ac{BI} offsets. In the next two sections we establish a formula for \ac{EMDT} under the assumption of a uniform distribution. We then proceed to establishing an upper bound on the makespan of \ac{EMDT}-optimal schedules, which will help us to evaluate performance of discovery algorithms proposed in Chapter~\ref{sec:strategies}.

\section{\acf{EMDT}}
\label{sec:emdt}

Since a device does not a priori know which network configurations are used by its neighbors, it is not possible to design a listening schedule that minimizes \ac{MDT}. However, in the absence of this information, a reasonable assumption for a device performing a discovery is that of a uniform distribution of configurations. It can then follow a discovery strategy that allows to minimize the expected value for \ac{MDT}, the \ac{EMDT}. Note that while mean value and expected value are synonyms, to improve readability, we use the term \acf{MDT} to refer to the mean value taken over neighbor networks in a particular environment (see Definition~\ref{def:mdt}), while we use the term \acf{EMDT} to refer to the expected value of \acp{MDT} taken over instances of network environments. 

To be more precise, we assume probabilities $P=\left(p_\kappa,\kappa\in K_{BC}\right)$ that a particular combination of channel, \ac{BI}, and offset are used by a neighbor network, where $p_\kappa$ are defined as follows. We assume that each channel and each \ac{BI} have equal probability to be selected by a neighboring network, and that all offsets feasible for a particular \ac{BI} have equal probability to be selected by a network using this \ac{BI}. Thus, a network configuration $\kappa\in K_{BC}$, $\kappa=\left(c_\kappa,b_\kappa,\delta_\kappa\right)$, has probability $p_\kappa=\frac{1}{b_\kappa\left|B\right|\left|C\right|}$ to be used by a neighbor network $\nu\in N$.

In the following proposition, we compute \ac{EMDT} and show that for uniform probabilities, \ac{EMDT} equals \ac{MDT} given a complete environment. 

\begin{restatable}{proposition}{mdtprop}
\label{prop:mdt}
For a set of networks $N$, a set of \acp{BI} $B\in\mathbb{F}_1$, a set of channels $C$, a complete listening schedule $\mathcal{L}$, and probabilities $p_\kappa=\frac{1}{b_\kappa\left|B\right|\left|C\right|}$ that a network configuration $\kappa\in K_{BC}$ is used by a network $\nu\in N$, \ac{EMDT} is given by
\begin{equation*}
E\left[\frac{1}{\left|N\right|}\sum_{\nu\in N}T_\nu\left(\mathcal{L}\right)\right]=\frac{1}{\left|B\right|\left|C\right|}\sum_{\kappa\in K_{BC}}\frac{1}{b_\kappa}T_\kappa\left(\mathcal{L}\right)\,.
\end{equation*}
Note that $\ac{EMDT}$ equals $\ac{MDT}$ for a complete environment
\end{restatable}
\begin{proof}
First, assume that the number of networks $\left|N\right|$ is fixed. 
Let $N_\kappa\subseteq N$ be the subset of networks using configuration $\kappa\in K_{BC}$. Observe that
\begin{align*}
E\left[\frac{1}{\left|N\right|}\sum_{\nu\in N}T_\nu\left(\mathcal{L}\right)\right]=
E\left[\sum_{\kappa\in K_{BC}}\frac{\left|N_\kappa\right|}{\left|N\right|}T_\kappa\left(\mathcal{L}\right)\right]=
\sum_{\kappa\in K_{BC}}\frac{E\left[\left|N_\kappa\right|\right]}{\left|N\right|}T_\kappa\left(\mathcal{L}\right)\,.
\end{align*}
In order to calculate $E\left[\left|N_\kappa\right|\right]$, we have to calculate $P\left[\left|N_\kappa\right|=n\right]$. Since we assumed $p_\kappa=\frac{1}{b_\kappa\left|B\right|\left|C\right|}$, we have
\begin{equation*}
P\left[\left|N_\kappa\right|=n\right]=\binom{\left|N\right|}{n}\frac{1}{\left(b_\kappa\left|B\right|\left|C\right|\right)^n}\cdot\left(1-\frac{1}{b_\kappa\left|B\right|\left|C\right|}\right)^{\left|N\right|-n}=
\binom{\left|N\right|}{n}\frac{\left(b_\kappa\left|B\right|\left|C\right|-1\right)^{\left|N\right|-n}}{\left(b_\kappa\left|B\right|\left|C\right|\right)^{\left|N\right|}}\,.
\end{equation*}
Consequently,

\begin{align*}
\frac{E\left[\left|N_\kappa\right|\right]}{\left|N\right|}&=
\frac{1}{\left|N\right|}\sum_{n=0}^{\left|N\right|}{n\cdot P\left[\left|N_\kappa\right|=n\right]}\\
&=\frac{1}{\left|N\right|}\sum_{n=0}^{\left|N\right|}{n\binom{\left|N\right|}{n}\frac{\left(b_\kappa\left|B\right|\left|C\right|-1\right)^{\left|N\right|-n}}{\left(b_\kappa\left|B\right|\left|C\right|\right)^{\left|N\right|}}}\\
&=\frac{1}{\left|N\right|}\left(\frac{b_\kappa\left|B\right|\left|C\right|-1}{b_\kappa\left|B\right|\left|C\right|}\right)^{\left|N\right|}\sum_{n=0}^{\left|N\right|}{n\binom{\left|N\right|}{n}\frac{1}{\left(b_\kappa\left|B\right|\left|C\right|-1\right)^n}}\\
&=\frac{1}{\left|N\right|}\left(\frac{b_\kappa\left|B\right|\left|C\right|-1}{b_\kappa\left|B\right|\left|C\right|}\right)^{\left|N\right|}\left|N\right|\frac{\left(b_\kappa\left|B\right|\left|C\right|\right)^{\left|N\right|-1}}{\left(b_\kappa\left|B\right|\left|C\right|-1\right)^{\left|N\right|}}\\
&=\frac{1}{b_\kappa\left|B\right|\left|C\right|}\,.
\end{align*}
Since the resulting expression does not depend on $\left|N\right|$, the assumption of fixed $\left|N\right|$ can be relaxed, and thus, the claim is proved.
\end{proof}

In the following section, we will use this proposition to study the number of time slots required to minimize \ac{EMDT}.

\section{Upper Bound on Makespan of \ac{EMDT}-Optimal Schedules}

In this section, we will show that for an arbitrary set of \acp{BI}, \ac{EMDT} can be minimized within $LCM(B)\cdot\left|C\right|$ time slots, and that, consequently, for \ac{BI} sets in $\mathbb{F}_2$, \ac{EMDT} can be minimized within $\max(B)\cdot\lvert C\rvert$ time slots, implying that in $\mathbb{F}_2$, \ac{EMDT}-optimal schedules are also makespan-optimal. The following proposition establishes an upper bound on the number of time slots required to minimize an arbitrary strictly increasing function of discovery times. The idea for the proof is to show that any schedule that results in a network being detected after time slot $LCM(B)\cdot\lvert C\rvert$ can be modified such that the network in question is detected before time slot $LCM(B)\cdot\lvert C\rvert$, without increasing the discovery times of other networks.

\begin{restatable}{proposition}{lcmBC}
\label{prop:lcmBC}
For an arbitrary set of \acp{BI} $B\in\mathbb{F}_1$, a set of channels $C$, and a function $f:\mathbb{N}^{\lvert K_{BC}\rvert}\rightarrow\mathbb{R}$, which is strictly increasing in each argument, complete schedules $\mathcal{L}^*$ that minimize $f\left(\left(T_\kappa\left(\mathcal{L}\right)\right)_{\kappa\in K_{BC}}\right)$ have a makespan $T_{\mathcal{L}^*}\leq LCM(B)\cdot\lvert C\rvert$.
\end{restatable}
\begin{proof}
Assume schedule $\mathcal{L}$ minimizes $f$ and $T_{\mathcal{L}}>LCM(B)\cdot\lvert C\rvert$. Consequently, there is at least one configuration $\kappa=\left(c,b,\delta\right)$ with discovery time $T_{\kappa}\left(\mathcal{L}\right)=T_{\mathcal{L}}>LCM(B)\cdot\lvert C\rvert$. Consider time slots $\tilde{\mathcal{T}}_\kappa=\left\{\delta+i\cdot LCM(B)\right\}_{i\in\left\{0,\ldots,\lvert C\rvert-1\right\}}$. Observe that $\left(\left\{c\right\}\times\tilde{\mathcal{T}}_\kappa\right)\cap\mathcal{L}=\emptyset$ since otherwise $\kappa$ would have been detected during one of the time slots in $\tilde{\mathcal{T}}_\kappa$. Consequently, there either exists an idle time slot $\tilde{t}\in\tilde{\mathcal{T}}_\kappa$, or, since $\lvert\tilde{\mathcal{T}}_\kappa\rvert=\lvert C\rvert$, there exist time slots $t',t''\in\tilde{\mathcal{T}}_\kappa$ and a channel $c'\neq c$ such that $\left(c',t'\right),\left(c',t''\right)\in\tilde{\mathcal{T}}_\kappa$.

In the first case, we construct a new schedule $\mathcal{L}'=\mathcal{L}\setminus\left\{\left(c,\,T_{\mathcal{L}}\right)\right\}\cup\left\{\left(c,\,\tilde{t}\right)\right\}$, such that $\kappa$ is detected during $\tilde{t}$ and none of the discovery times of other network configurations are increased.

In the second case, we construct a new schedule $\mathcal{L}'=\mathcal{L}\setminus\left\{\left(c,\,T_{\mathcal{L}}\right),\left(c',\,t''\right)\right\}\cup\left\{\left(c,\,t''\right)\right\}$. With the new schedule, configuration $\kappa$ is detected during time slot $t''$. In order to show that the discovery times of other networks do not increase, consider the function $\delta(t)$, defined in Chapter~\ref{sec:system}, providing for each time slot $t$ a vector of offsets that can be detected in $t$. Since periodicity of $\delta(t)$ is $LCM(B)$, we conclude that $\delta(t')=\delta(t'')$, and thus no discoveries are performed during time slot $t''$ with the schedule $\mathcal{L}$. Consequently, none of the discovery times are increased with the new schedule.

Repeating the above procedure for each $\kappa$ with discovery time $T_{\kappa}\left(\mathcal{L}\right)>LCM(B)\cdot\lvert C\rvert$ results in a schedule $\mathcal{L}^*$ with makespan $T_{\mathcal{L}^*}\leq LCM(B)\cdot\lvert C\rvert$ with $f\left(\left(T_\kappa\left(\mathcal{L}^*\right)\right)_{\kappa\in K_{BC}}\right)<f\left(\left(T_\kappa\left(\mathcal{L}\right)\right)_{\kappa\in K_{BC}}\right)$, proving the claim.
\end{proof}

The following Corollary presents a notable consequence from Proposition~\ref{prop:lcmBC} for \ac{BI} sets from $\mathbb{F}_2$.

\begin{restatable}{corollary}{lcmBCcor1}
\label{cor:lcmBC}
For a \ac{BI} set $B\in\mathbb{F}_2$, a set of channels $C$, and a function $f:\mathbb{N}^{\lvert K_{BC}\rvert}\rightarrow\mathbb{R}$, which is strictly increasing in each argument, complete schedules $\mathcal{L}^*$ that minimize \newline $f\left(\left(T_\kappa\left(\mathcal{L}\right)\right)_{\kappa\in K_{BC}}\right)$ are makespan-optimal.
\end{restatable}
\begin{proof}
The claim follows directly from Corollary~\ref{cor:maxBC}, Proposition~\ref{prop:lcmBC}, and the defining property of $B\in\mathbb{F}_2$ that $LCM(B)=\max(B)$.
\end{proof}

Please observe that Proposition~\ref{prop:mdt} implies that the upper bounds established in Proposition~\ref{prop:lcmBC} and Corollary~\ref{cor:lcmBC} also apply to schedules minimizing \ac{EMDT}.

\begin{restatable}{corollary}{lcmBCcor1mdt}
\label{cor:mdtcor1}
For an arbitrary set of \acp{BI} $B\in\mathbb{F}_1$, and a set of channels $C$, \ac{EMDT}-optimal listening schedules $\mathcal{L}^*$ have a makespan $T_{\mathcal{L}^*}\leq LCM(B)\cdot\lvert C\rvert$.
\end{restatable}
\begin{proof}
From Proposition~\ref{prop:mdt} we obtain an expression for \ac{EMDT}, which is strictly increasing in each configuration detection time $T_\kappa$. Applying Proposition~\ref{prop:lcmBC} we obtain the claim.
\end{proof}

\begin{restatable}{corollary}{lcmBCcor2}
\label{cor:mdtcor2}
For a set of \acp{BI} $B\in\mathbb{F}_2$, and a set of channels $C$, \ac{EMDT}-optimal listening schedules are also makespan-optimal.
\end{restatable}
\begin{proof}
From Proposition~\ref{prop:mdt} we obtain an expression for \ac{EMDT}, which is strictly increasing in each configuration detection time $T_\kappa$. Applying Corollary~\ref{cor:lcmBC} we obtain the claim.
\end{proof}

In the following Chapter we will propose efficient approaches to computing listening schedules that are both \ac{EMDT}-optimal and makespan-optimal.

\chapter{Discovery Strategies}
\label{sec:strategies}

In this section we present several novel discovery algorithms for multichannel environments, that are \acf{EMDT}-optimal and makespan-optimal for the family of \ac{BI} sets $\mathbb{F}_3$ (see Chapter~\ref{sec:BIFamilies} for a definition). Note that $\mathbb{F}_3$ contains a broad spectrum of \ac{BI} sets, completely including \ac{BI} sets defined by the IEEE~802.15.4 standard, and a large fraction of \ac{BI} sets defined by the IEEE~802.11 standard.
In addition, we develop an \ac{ILP}-based approach, denoted GENOPT, that computes \ac{EMDT}-optimal  discovery schedules for arbitrary \ac{BI} sets in $\mathbb{F}_1$.

Note that all presented strategies operate by passively listening to periodically transmitted beacon messages of neighbors that are agnostic to the discovery process. Moreover, due to their optimality for quite general sets of \acp{BI}, such as $\mathbb{F}_1$, $\mathbb{F}_2$, or $\mathbb{F}_3$, they can be deployed without conflicting with existing \ac{MAC} layer technologies, allowing network operators to select \ac{BI} sets suited best for the targeted application and/or device characteristics, such as, e.g., energy constraints. In particular, our results apply to \ac{BI} sets used by technologies such as IEEE~802.11 and~IEEE~802.15.4.

%


The rest of this chapter is structured as follows. In Section~\ref{sec:disc_154}, we briefly present discovery strategies specified by the IEEE~802.15.4 standard as well as our previous work on optimized IEEE~802.15.4 discovery. In Section~\ref{sec:greedy}, we define the family of algorithms \ALG{}, prove their optimality, and study their complexity. In Section~\ref{sec:chantrain}, we describe a strategy named CHAN TRAIN, which is a modification of the \ALG{} approach attempting at reducing the number of channel switches. An \ac{ILP}-based approach to computing \ac{EMDT}-optimal listening schedules for arbitrary \ac{BI} sets is presented in Section~\ref{sec:genopt}. Finally, in Section~\ref{sec:opt2} we describe a discovery strategy for the special case of \ac{BI} sets with two elements.

\section{Discovery Strategies for IEEE~802.15.4 Networks}
\label{sec:disc_154}

In the following we will describe discovery strategies for IEEE~802.15.4 networks. First, we will explain the different scanning techniques offered in the IEEE~802.15.4 standard and then we will give a brief overview of our previous work on optimized discovery strategies for IEEE~802.15.4 networks.

\subsection{PSV - Passive Scan in IEEE~802.15.4}
\label{subsec:psv}

The IEEE~802.15.4 standard supports \acp{BI} that are multiple of powers of 2. We denote the corresponding family of \ac{BI} sets by $\mathbb{F}_{IEEE~802.15.4}$ (see also Chapter~\ref{sec:BIFamilies}). The IEEE~802.15.4 standard defines four types of scanning techniques~\cite{ieee802154}. When performing an active or an orphan scan, devices transmit beacon requests or orphan notifications on each selected channel. With the passive scan and the energy detection scan, a device only listens on channels without any transmission. Frames are only decoded in the passive scan while the result of an energy scan is the peak energy per channel.

In our work we focus on passive discovery techniques and therefore compare our strategies to the passive scan of IEEE~802.15.4, denoted by PSV. PSV proceeds by sequentially listening on each channel $c\in C$ for $\max(B)$ time slots. Note that listening schedules generated by PSV are optimal w.r.t. the makespan and the number of channel switches for arbitrary \ac{BI} sets $B \in \mathbb{F}_1$. Still, they fail to minimize \ac{EMDT}, meaning that although the time when the last network is detected is minimized, the time of discovery of other network can be significantly worse than its optimal value, as shown in the following example. Consider a setting with $B=\{1,2\}$ and $\lvert C\rvert=2$. The schedule generated by PSV has an \ac{EMDT} of 2.25, while minimum \ac{EMDT} is 2 (scanning channel 1 during time slots 1 and 4, and channel 2 during time slots 2 and 3). The optimality gap further increases for larger scenarios.

\subsection{(SW)OPT - Previous Work on Optimized IEEE~802.15.4 Discovery Strategies}
\label{subsec:swopt}

In our previous work~\cite{Karowski11, Karowski13} we developed discovery strategies for IEEE~802.15.4 networks, the OPTimzed (OPT) and the SWitched OPTimized (SWOPT) strategy. Both strategies are based on solving \ac{ILP} minimizing the \ac{EMDT} and showed significant improvement as compared to PSV as well as the SWEEP strategy~\cite{willig10}. OPT and SWOPT use the same \ac{ILP}, however, SWOPT additionally performs the preprocessing on the \ac{BI} set described in Section~\ref{sec:BIFamilies} (dividing all \acp{BI} by their \ac{GCD}), resulting in listening schedules with fewer channel switches. Note that OPT and SWOPT generate recursive listening schedules. 

Even though (SW)OPT was developed for \ac{BI} sets from $\mathbb{F}_{\textrm{IEEE 802.15.4}}$, it can be adapted to support \ac{BI} sets from $\mathbb{F}_{3}$. However, due to a grouping of time slots in the formulation of the \ac{ILP} that allows to compute the discovery probability of individual network configurations, (SW)OPT cannot be extended to support \ac{BI} sets from $\mathbb{F}_{2}$.



\section{\ALG{} Discovery Algorithms}
\label{sec:greedy}

In the following, we show that for the family of \ac{BI} sets $\mathbb{F}_3$, schedules that are both \ac{EMDT}-optimal and makespan-optimal can be computed in a very efficient way by greedily maximizing the number of discovered configurations in each time slot. In particular, we will show that for $B\in\mathbb{F}_3$ the computational complexity required by a straightforward implementation is $\mathcal{O}\left(\left|C\right|^2\left|B\right|\max(B)\right)$, while memory consumption is $\mathcal{O}\left(\left|C\right|\left|B\right|\sum_{b\in B}b\right)$. We start by formally defining the family of algorithms \ALG{}.

\begin{restatable}[\ALG{}]{definition}{mindy}
\label{def:mindy}
The family of algorithms \ALG{} contains all algorithms that greedily maximize the number of detected configurations in each time slot. More precisely, for an arbitrary \ac{BI} set $B\in\mathbb{F}_1$ and channel set $C$, in every time slot $t$, an algorithm $A\in\ALG{}$ scans channel $c\in C$ that maximizes the number of configurations from $K_{BC}$ discovered in time slot $t$, given the configurations discovered in previous time slots $\left\{1,\ldots,t-1\right\}$.
\end{restatable}
In the following, we will say "\ALG{} listening schedule" to refer to a schedule generated by \ALG{} algorithm.

Note that for arbitrary \ac{BI} sets, a \ALG{} algorithm does not necessarily generate \ac{EMDT}-optimal or makespan-optimal schedules. Consider an example using $B = \{1,2,3,5\}\in\mathbb{F}_1\setminus\mathbb{F}_2$ and $|C| = 3$. Figure~\ref{fig:GREEDY_F1_non_makespan_MDT} shows a \ALG{} listening schedule, which is neither makespan-optimal nor \ac{EMDT}-optimal, as can be seen from a comparison with the listening schedule in Figure~\ref{fig:GREEDY_F1_non_makespan_MDT_GENOPT_comp}, which is both makespan-optimal and \ac{EMDT}-optimal. The number in each square represents the number of configurations that can be discovered by scanning a particular channel during a particular time slot. Gray squares represent channels scanned by the illustrated listening schedule. In particular, optimal value for the makespan in this example constitutes 15 time slots, while optimal value for \ac{EMDT} is $4.875$ time slots, in contrast to $5.125$ time slots achieved by the \ALG{} schedule. Another example using $B = \{2,3,4,6,12\} \in \mathbb{F}_2 \setminus \mathbb{F}_3$ and $|C| = 2$ is depicted in Figure~\ref{fig:GREEDY_non_makespan_MDT_opt_F2}. The \ALG{} listening schedule achieves a \ac{EMDT} of 6.3 time slots in comparison with the optimal value of 6.1 time slots.

\begin{figure*}
\begin{center}
		\subfloat[GREEDY]{		
        \includegraphics[width=1.0\textwidth]{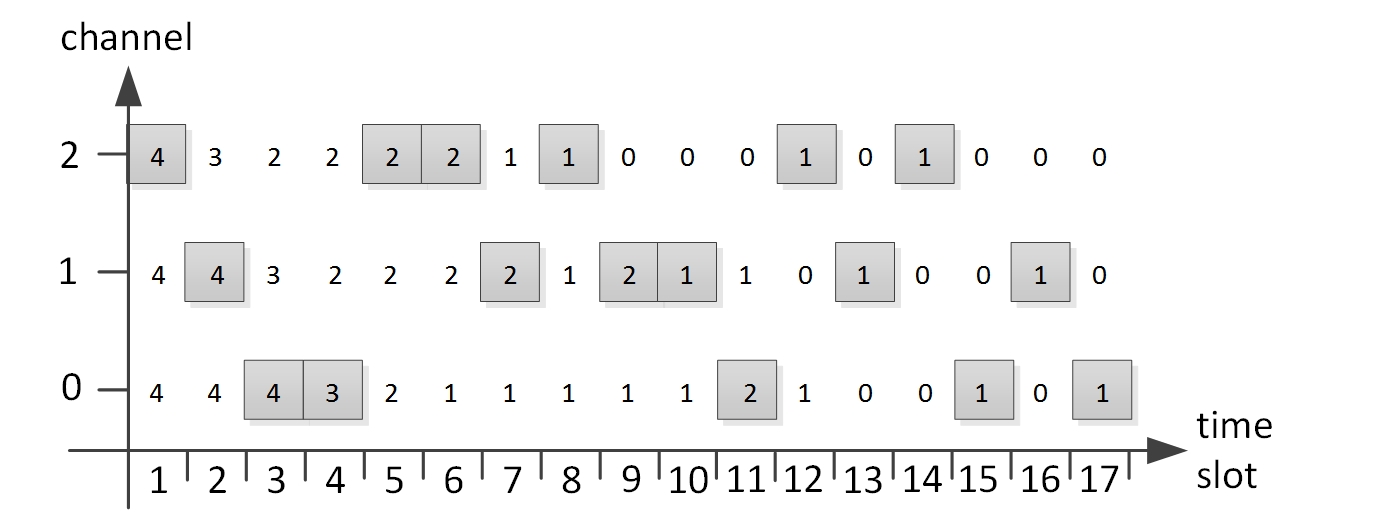}
				\label{fig:GREEDY_F1_non_makespan_MDT}
		} \hfill
		\subfloat[\ac{EMDT}-optimal schedule]{
        \includegraphics[width=1.0\textwidth]{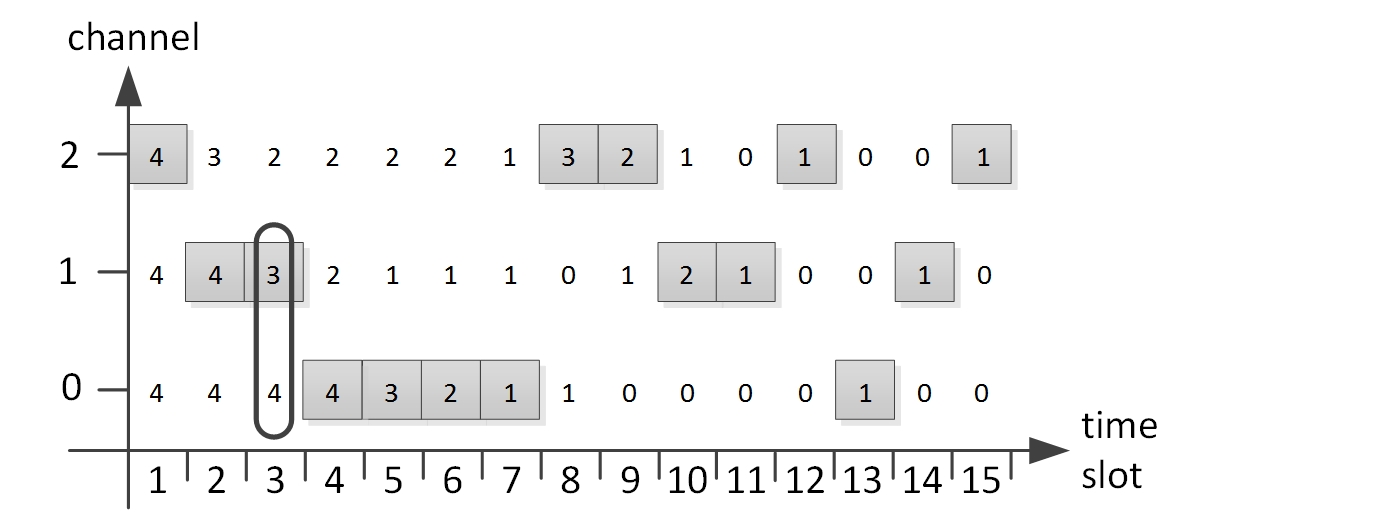}
				\label{fig:GREEDY_F1_non_makespan_MDT_GENOPT_comp}
    }
\caption{Example for a \ALG{} listening schedule depicted by gray boxes that is neither \ac{EMDT} nor makespan optimal using \ac{BI} set $B = \{1,2,3,5\} \not\in \mathbb{F}_2$ and $|C| = 3$.}
 \label{fig:GREEDY_non_makespan_MDT_opt_F1}
\end{center}
\end{figure*}

\begin{figure*}
\begin{center}
		\subfloat[GREEDY]{		
        \includegraphics[width=1.0\textwidth]{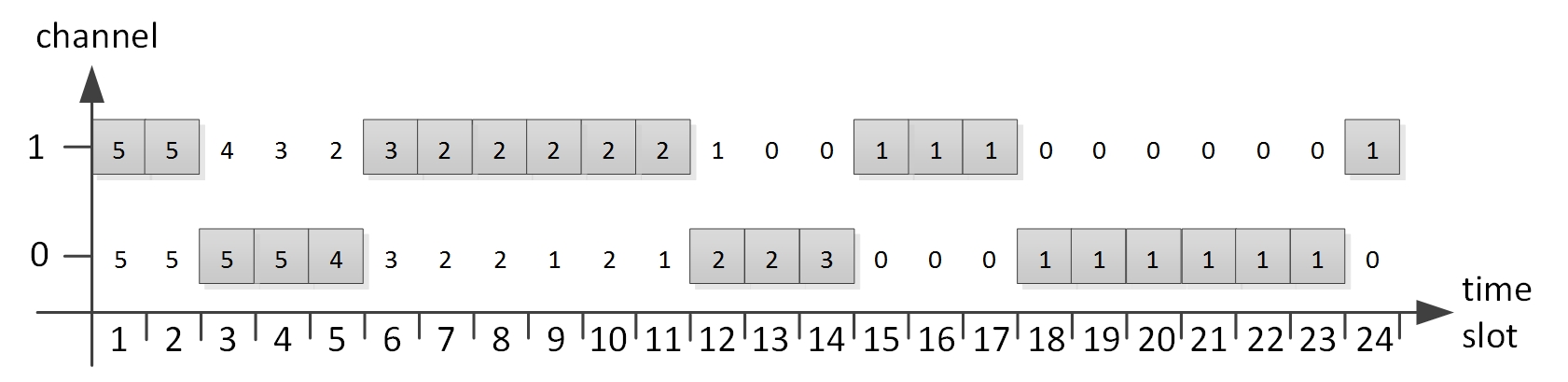}
				\label{fig:GREEDY_F2_non_MDT}
		} \hfill
		\subfloat[\ac{EMDT}-optimal schedule]{
        \includegraphics[width=1.0\textwidth]{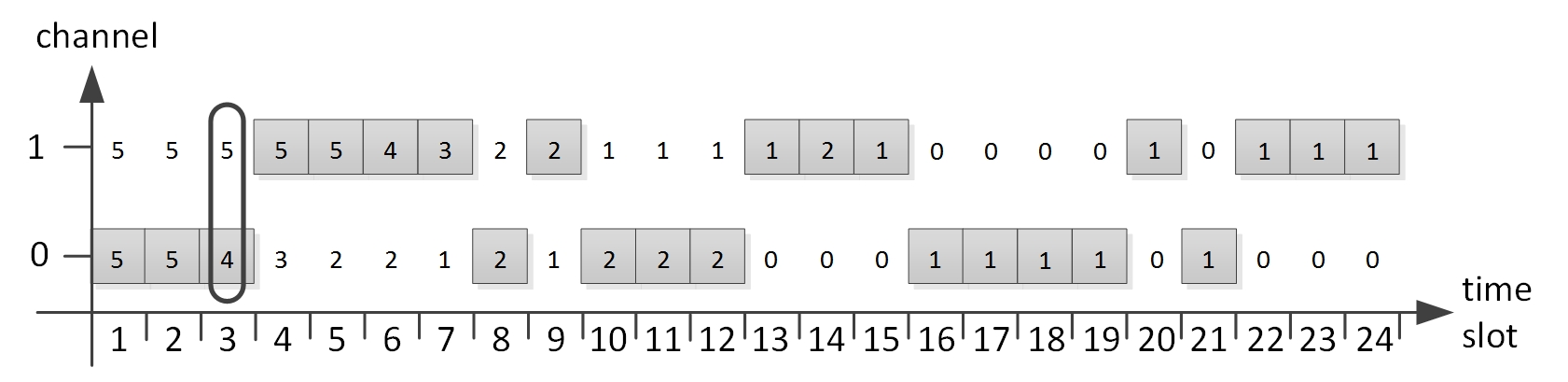}
				\label{fig:GREEDY_F2_non_MDT_GENOPT}
    }
\caption{Example for a \ALG{} listening schedule depicted by gray boxes that is not \ac{EMDT} optimal using \ac{BI} set $B = \{2,3,4,6,12\} \in \mathbb{F}_2 \setminus \mathbb{F}_3$ and $|C| = 2$.}
 \label{fig:GREEDY_non_makespan_MDT_opt_F2}
\end{center}
\end{figure*}

In the following, however, we will show that for the family of \ac{BI} sets $\mathbb{F}_3$, \ALG{} algorithms always result in \ac{EMDT}-optimal and makespan-optimal schedules. We prove this claim by showing that \ALG{} algorithms generate recursive schedules (see Definition~\ref{def:recursive_schedule}). Afterwards, we will discuss computational complexity of \ALG{} algorithms.


\begin{restatable}{proposition}{mindynumdetections}
\label{prop:mindynumdetections}
In network environments using \ac{BI} sets from $\mathbb{F}_3$, a schedule is \ALG{} if and only if it is recursive.
\end{restatable}
\begin{proof}
W.l.o.g. we assume $B=\left\{b_1,\ldots,b_n\right\}$ with $b_i<b_j$ for $i<j$. From Proposition~\ref{prop:recursive_characterization} we know that a schedule is recursive if and only if in each time slot $t\in\left[1,\,b_i\cdot\lvert C\rvert\right]$, for a $b_i\in B$, it discovers at least $n-i+1$ configurations, which is the maximum number of discoverable configurations in each time slot. From Proposition~\ref{prop:recursive_F3} we know that in a complete environment with a \ac{BI} set $B\in\mathbb{F}_3$ a recursive schedule exists. Consequently, since by definition, \ALG{} schedules maximize the number of configurations discovered in each time slot, \ALG{} schedules are recursive. The claim that a recursive schedule is \ALG{} follows directly from Corollary~\ref{cor:recursive_optimal}.
\end{proof}

From this result we are able to derive the following conclusions.

\begin{restatable}{corollary}{mindymakespan}
\label{cor:mindymakespan}
In a network environment with a \ac{BI} set $B\in\mathbb{F}_3$, \ALG{} schedules are complete, makespan-optimal, \ac{EMDT}-optimal, and maximize the number of configurations detected until each time slot, for each $B'\subset B$.
\end{restatable}
\begin{proof}
This is a direct consequence of Corollary~\ref{cor:recursive_optimal}, and Propositions~\ref{prop:mdt} and~\ref{prop:mindynumdetections}.
\end{proof}

As shown in the example above, the assumption of a \ac{BI} set from $\mathbb{F}_3$ is crucial to prove \ac{EMDT}-optimality of \ALG{} schedules. However, relaxing the assumption on the used \ac{BI} set, we still can show that for the family of \ac{BI} sets $\mathbb{F}_2$, \ALG{} schedules are makespan-optimal, as stated in the following Proposition.

\begin{restatable}{proposition}{greedyf2makespanopt}
\label{prop:greedyf2makespanopt}
In a network environment with a \ac{BI} set $B\in\mathbb{F}_2$, \ALG{} schedules are makespan-optimal.
\end{restatable}
\begin{proof}
Consider \ac{BI} set $B=\left\{b_1,\ldots,b_n\right\}\in\mathbb{F}_2$, with $b_i<b_j$ for $i<j$, and a channel set $C$. Consider a \ALG{} listening schedule $\mathcal{L}$. We first show that in each time slot $t\leq \max(B)\cdot\lvert C\rvert$, $\mathcal{L}$ detects exactly one configuration $\left(c,b_n,\delta_{b_n}(t)\right)$, for a $c\in C$. Assume that this is not the case, that is, at time slot $t$, $\mathcal{L}$ scans channel $c$ such that $\left(c,b_n,\delta_{b_n}(t)\right)$ have already been detected earlier in a time slot $t-\ell\cdot b_n$ for an $\ell>0$. Then, however, all other configurations $\left\{\left(c,b_i,\delta_{b_i}(t)\right)\right\}_{i\in\{1,\ldots,n\}}$ has been detected earlier as well, since in $\mathbb{F}_2$, periodicity of $\delta_{b_n}(t)$ is an integer multiple of periodicity of $\delta_{b_i}(t)$ for $i\in\{1,\ldots,n-1\}$. Thus, $\mathcal{L}$ would scan a time slot resulting in 0 detected configurations or this time slot would be idle. On the other hand, however, we know from Corollary~\ref{cor:maxBC} that for each time slot $t\leq b_n \cdot \lvert C\rvert$ there exists a channel $c\in C$ such that a configuration $\left(c,b_n,\delta_{b_n}(t)\right)$ can be detected. This contradicts the assumption that $\mathcal{L}$ is \ALG{}, proving the claim.
\end{proof}

Note that \ALG{} algorithms have a compelling property of low complexity. A straightforward example implementation proceeds as follows. It iterates over time slots until all configurations are discovered. For each configuration it stores a binary variable indicating if it has been discovered or not, resulting in $\left|C\right|\left|B\right|\sum_{b\in B}b$ bits of required memory space. At each time slot $t$, it iterates over all channels $c\in C$, computing for each channel, which of configurations $\left\{\left(c,b,\delta_b(t)\right)\,|\,b\in B\right\}$ are not yet discovered. Finally, it selects a channel, for which this number is highest. Consequently, computational complexity at each time slot is $\mathcal{O}\left(\left|C\right|\left|B\right|\right)$. The overall computational complexity depends on the number of time slots required to discover all configurations. Since for complete environments with $B\in\mathbb{F}_2$ \ALG{} algorithms are makespan-optimal, the number of time slots is $\max(B)\left|C\right|$, resulting in total complexity of $\mathcal{O}\left(\left|C\right|^2\left|B\right|\max(B)\right)$. In $\mathbb{F}_1$, we only know that \ALG{} algorithms terminate at latest in time slot $LCM(B)\lvert C\rvert$. (Proof is similar to proof of Proposition~\ref{prop:greedyf2makespanopt}.) That is, an upper bound for the computational complexity over $\mathbb{F}_1$ is $\mathcal{O}\left(\left|C\right|^2\left|B\right|LCM(B)\right)$.

Please note that if the assumption made in our system model in Section~\ref{sec:system} that the \ac{GCD} of considered \ac{BI} sets is 1 does not hold, the complexity can be further reduced by replacing $B$ by $B'=\left\{\frac{b}{GCD(B)}\right\}_{b\in B}$. This preprocessing step allows to reduce computational complexity over $\mathbb{F}_2$ to $\mathcal{O}\left(\left|C\right|^2\left|B\right|\frac{\max(B)}{GCD(B)}\right)$, over $\mathbb{F}_1$ the upper bound becomes $\mathcal{O}\left(\left|C\right|^2\left|B\right|\frac{LCM(B)}{GCD(B)}\right)$.

Note that in general, in each time slot, there might exist several channels whose selection maximizes the number of discovered configurations. Therefore, \ALG{} is a family of algorithms and not a single algorithm. The difference between them is the tiebreaking rule that, in each time slot, selects one channel to be scanned from the set of channels maximizing the number of discovered configurations. In the following we describe two deterministic and two probabilistic tiebreaking rules.

\vspace{0.1cm}
\textbf{\GreedyRnd{}} randomly selects a channel among the channels maximizing the number of discovered configurations.

\vspace{0.1cm}
\textbf{\GreedyDeter{}} selects the channel with the highest identifier among the channels maximizing the number of discovered configurations.

\vspace{0.1cm}
\textbf{\GreedyTrainRnd{}} tests if the channel scanned in the previous time slot is within the set of channels maximizing the number of discovered configurations. If yes, it is selected. If no, it proceeds as \GreedyRnd{}. By prioritizing the most recently selected channel, \GreedyTrainRnd{} tries to reduce the number of channel switches by creating a channel train.

\vspace{0.1cm}
\textbf{\GreedyTrainDeter{}} is similar to \GreedyTrainRnd{} but without a random component. It tests if the channel scanned in the previous time slot is within the set of channels maximizing the number of discovered configurations. If yes, it is selected. If no, it proceeds as \GreedyDeter{}.

\section{CHAN TRAIN - Reducing the Number of Channel Switches}
\label{sec:chantrain}

In this section, we propose an algorithm, which, similar to the \GreedyTrain{} algorithms in the previous section, in addition to minimizing \ac{EMDT}, tries to minimize the number of channel switches. While \GreedyTrainRnd{} and \GreedyTrainDeter{} do that by taking into account the most recently scanned channel, the CHAN TRAIN strategy goes one step further and also takes into account channels that will be scanned in future time slots.

To be more precise, in time slot $t=1$, CHAN TRAIN computes the set of channels maximizing the number of detected configurations in $t$. Out of those, it than selects the channel maximizing the sum of consecutive future time slots $t'$ where at least the same number of configurations can be detected and of the number of previous consecutive time slots allocated on this channel. In case multiple channels achieve the same maximum sum, it selects the channel with the lowest identifier.
It then jumps to $t+t'$ and repeats the procedure.

\begin{restatable}{proposition}{chantraingreedy}
\label{prop:chantraingreedy}
In an environment with a \ac{BI} set $B\in\mathbb{F}_3$, CHAN TRAIN is \ALG{}.
\end{restatable}
\begin{proof}
From Proposition~\ref{prop:recursive_characterization} and the definition of CHAN TRAIN we conclude that CHAN TRAIN generates recursive schedules for \ac{BI} sets from $\mathbb{F}_3$. Propositions~\ref{prop:mdt} and~\ref{prop:mindynumdetections} then implies the claim.
\end{proof}

Note that for a \ac{BI} set $B\not\in\mathbb{F}_3$ Proposition~\ref{prop:chantraingreedy} is no longer true, as illustrated by the example in Figure~\ref{fig:CHAINTRAIN_nonGreedyBehaviour}.


\begin{figure*}
\begin{center}
\centerline{
		\subfloat[CHAN TRAIN]{		
        \includegraphics[width=0.5\textwidth]{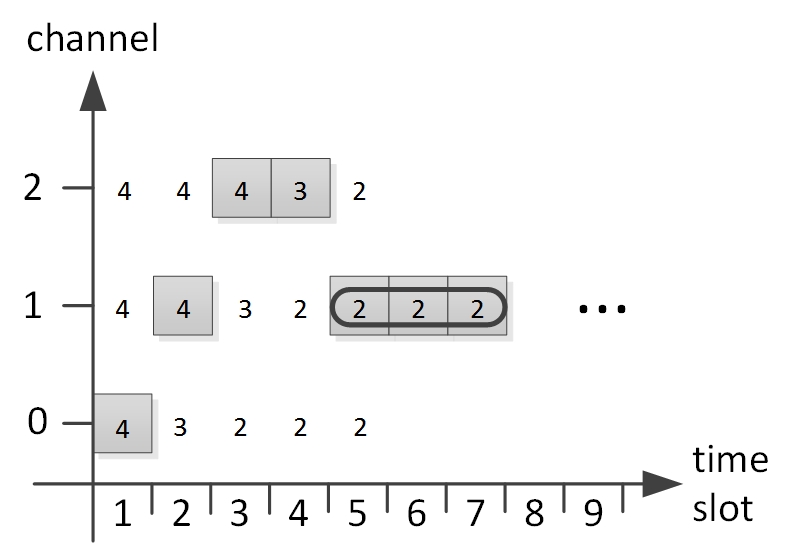}
				\label{fig:CHANTRAIN_F2}
		}
		\hspace{\evalHspace}
		\subfloat[GREEDY]{
        \includegraphics[width=0.5\textwidth]{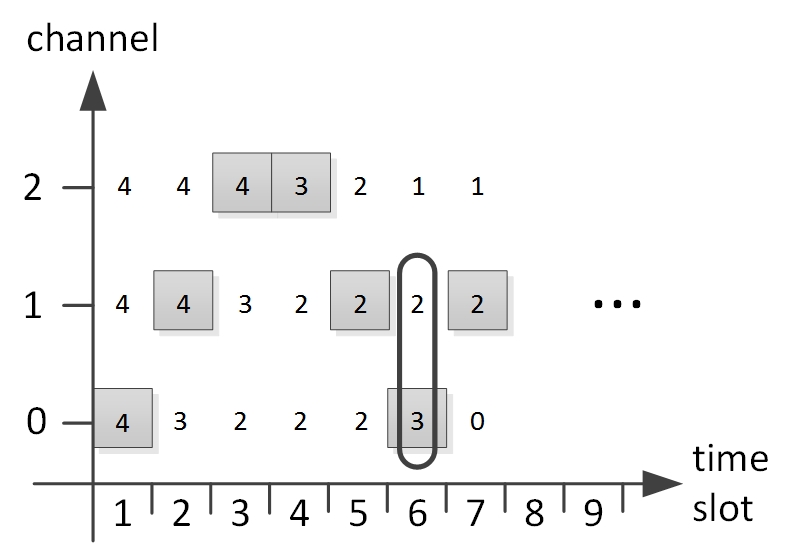}
				\label{fig:CHANTRAIN_GREEDY_F2}
    }
}
\caption{Example showing the non-greedy behavior of the CHAN TRAIN strategy for the BI set $B = \{1,2,3,6\} \in \mathbb{F}_2\setminus\mathbb{F}_3$ and $|C| = 3$}
 \label{fig:CHAINTRAIN_nonGreedyBehaviour}
\end{center}
\end{figure*}

\section{GENOPT - Minimizing \acl{MDT} for Arbitrary \acl{BI} Sets}
\label{sec:genopt}

In the previous sections, we presented a family of low-complexity algorithms that are \ac{EMDT}-optimal for the broad family of \ac{BI} sets $\mathbb{F}_3$. Still, those algorithms might fail to achieve optimality for \ac{BI} sets in $\mathbb{F}_1\setminus\mathbb{F}_3$, as shown in examples in Figures~\ref{fig:GREEDY_non_makespan_MDT_opt_F1} and~\ref{fig:GREEDY_non_makespan_MDT_opt_F2}.

In this section, we formulate a set of linear constraints describing a complete listening schedule for arbitrary \ac{BI} sets $B\in\mathbb{F}_1$. We use this set to develop a discovery strategy GENOPT that minimizes \ac{EMDT}, by complementing it with an appropriate linear objective function. Note that this set of constraints might be used to generate schedules optimized w.r.t. other metrics. Since GENOPT involves solving an \ac{ILP}, it suffers from high computational complexity and memory consumption, and should only be performed offline and for network environments that are limited in size. 

Note that listening schedules generated by GENOPT might not be makespan-optimal as shown in following example.
Figure~\ref{fig:GENOPT_F1_non_makespan_opt} depicts an example using $B = \{1,2,4,5\}$ and $|C| = 2$ for which an \ac{EMDT}-optimal listening schedule cannot be constructed within $\max(B)\cdot |C|$ time slots.
Figure~\ref{fig:GENOPT_F1_makespan_tmax-optMakespan} shows an \ac{EMDT}-optimal listening schedule with the additional constraint of using at most $\max(B) \cdot |C|$ time slots. 
Observe that its \ac{EMDT} is 3.875 time slots as compared to the optimal value of 3.75 time slots of the schedule shown in Figure~\ref{fig:GENOPT_F1_makespan_tmax-lcm}. 

\begin{figure*}
\begin{center}
		\subfloat[Makespan constrained by the optimum value $\max(B) \cdot |C|$]{		
        \includegraphics[width=0.6\textwidth]{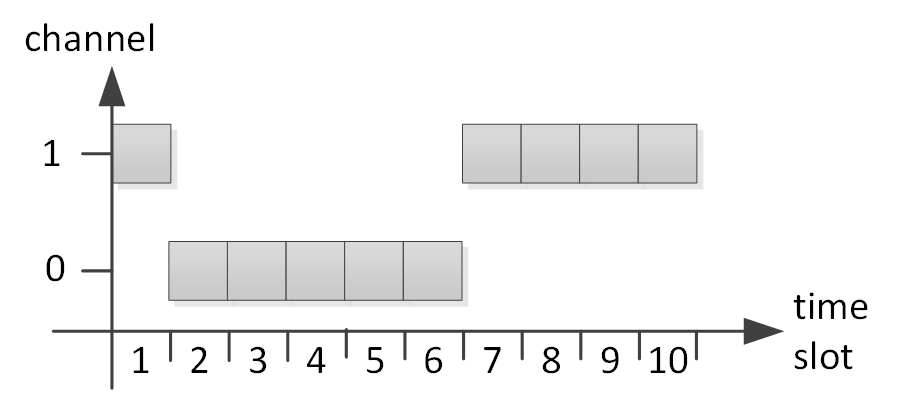}
				\label{fig:GENOPT_F1_makespan_tmax-optMakespan}
		} \hfill
		\subfloat[Makespan constrained by $LCM(B) \cdot |C|$]{
        \includegraphics[width=0.6\textwidth]{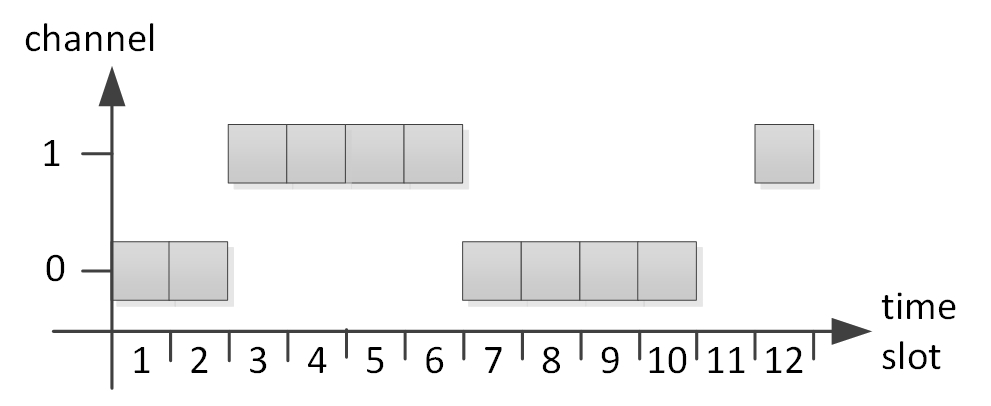}
				\label{fig:GENOPT_F1_makespan_tmax-lcm}
    }
\caption{Example for a listening schedule generated by GENOPT, using \ac{BI} set $B = \{1,2,4,5\} \in \mathbb{F}_1 \setminus \mathbb{F}_2$ and $|C| = 2$. Observe that imposing an upper bound on the makespan of the generated schedule increases \ac{EMDT}. Therefore, in this example, no schedule exists which is both \ac{EMDT}-optimal and makespan-optimal.}
 \label{fig:GENOPT_F1_non_makespan_opt}
\end{center}
\end{figure*}


To formulate GENOPT, we define the following variables. 

\begin{eqnarray*}
\begin{aligned}
x_{c,t,b} &= \begin{cases}
				\begin{aligned}
		   \textrm{1  ,} & \textrm{ if configuration } \left(c,b,\delta_t(b)\right) \textrm{ is detected during scan of channel } c \textrm{ in time slot } t \\
           \textrm{0  ,} & \textrm{ otherwise}  
				\end{aligned}					
        \end{cases} \\
h_{c,t} &= \begin{cases}
				\begin{aligned}
           \textrm{1  ,} & \textrm{ if a scan is performed on channel } c \textrm{ at time slot } t \\
           \textrm{0  ,} & \textrm{ otherwise}  
				\end{aligned}	
        \end{cases}
\end{aligned}
 \end{eqnarray*}

GENOPT can be formulated as follows.

\begin{align*}
 \min\quad     &\frac{1}{\left|C\right|\left|B\right|} \sum_{c\in C}\sum_{b\in B}\sum_{t=1}^{LCM(B)\left|C\right|}{x_{c,t,b}\cdot t\cdot\frac{1}{b}} &\notag   \\
 \text{s.t.}   &\sum_{m=0}^{\frac{LCM(B) \cdot \left|C\right|}{b} - 1}{x_{c,mb+\delta,b}}=1 \qquad\text{for all}\;c\in C,\;b\in B,\;\delta\in\left\{1,\ldots,b\right\}  &\label{eq:C1}\tag{C1}\\
							 &x_{c,b,t}\leq h_{c,t} 																											\hspace{3.1cm}\text{for all}\;c\in C,\;b\in B,\;t\in\left\{1,\ldots,LCM(B)\left|C\right|\right\} & \label{eq:C2} \tag{C2} \\
							 &\sum_{c\in C}h_{c,t}\leq 1 																								  \hspace{3cm}\text{for all}\;t\in\left\{1,\ldots,LCM(B)\left|C\right|\right\}\,.  \label{eq:C3} &\tag{C3}
\end{align*}

In this formulation, constraint~\eqref{eq:C1} ensures that each configuration is detected,~\eqref{eq:C2} ensures that a configuration $\left(c,b,\delta_t(b)\right)$ can only be detected if channel $c$ is scanned during time slot $t$,~\eqref{eq:C3} makes sure that only one channel is scanned during a time slot.

\section{\OPTBTwo{} - Special case: $\lvert B\rvert=2$}
\label{sec:opt2}

In this section, we describe the discovery strategy \OPTBTwo{} for the special case where the \ac{BI} set contains two elements, $B=\left\{b_1,\,b_2\right\}$, $b_1<b_2$, which is \ac{EMDT}-optimal for any $B\in\mathbb{F}_1$ with $\lvert B\rvert=2$. The number of channels can be arbitrarily large.

\begin{restatable}[\OPTBTwo{}]{definition}{opt2}
\label{def:opt2}
We are given a \ac{BI} set $B=\left\{b_1,\,b_2\right\}$, with $b_1<b_2$, and a set of channels $C=\left\{c_1,\ldots,c_m\right\}$. Discovery strategy \OPTBTwo{} generates a listening schedule, where channel $j\in\{1,\ldots,m\}$ is scanned during time slots $\left[(j-1)b_1+1,\,jb_1\right]$ and
\newline$\left[mb_1+(m-j)(b_2-b_1)+1,\,mb_1+(m-j+1)(b_2-b_1)\right]$.
\end{restatable}

\begin{restatable}[]{proposition}{propopt2}
\label{prop:opt2}
In network environments with \ac{BI} sets $B=\left\{b_1,\,b_2\right\}$, \OPTBTwo{} is \ac{EMDT}-optimal and makespan-optimal.
\end{restatable}
\begin{proof}
Observe that \OPTBTwo{} generates recursive schedules. The claim then follows from Corollary~\ref{cor:recursive_optimal} and Proposition~\ref{prop:mdt}.
\end{proof}

\section{Results Overview}
Table~\ref{tab:strategy_overview} provides an overview of optimality results for studied discovery strategies w.r.t. \ac{EMDT}, makespan and number of channel switches, for different \ac{BI} families, accompanied by computational complexity. Note that if a discovery strategy is makespan-optimal the generated listening schedules are also energy-optimal, since they do not contain idle time slots or redundant scans. The performance metrics we use for evaluation are described in detail in Section~\ref{sec:performance_metrics}.

\begin{table}[h]
\begin{center}
\begin{tabular}{lcccl}
\hline
Strategy& Makespan		& \ac{EMDT} 				& Channel Switches		 & Complexity 	\\\hline
PSV 		&	$\mathbb{F}_1$	&									& $\mathbb{F}_1$	 	 & $\mathcal{O}\left(C\right)$\\
(SW)OPT & $\mathbb{F}_3$	& $\mathbb{F}_3$	& 								 	 & High (\ac{ILP})\\
GREEDY 	& $\mathbb{F}_2$ 	& $\mathbb{F}_3$	&										 & $\mathcal{O}\left(\left|C\right|^2\left|B\right|LCM(B)\right)$	\\
GENOPT 	& $\mathbb{F}_2$	& $\mathbb{F}_1$	&									   & High (\ac{ILP})\\
\OPTBTwo{} (only $|B|=2$) & $\mathbb{F}_1$   & $\mathbb{F}_1$  & & $\mathcal{O}\left(C\right)$\\
 \hline
\end{tabular}
\end{center}
\caption{Overview of discovery strategies and their optimality w.r.t. the makespan, \ac{EMDT}, and number of channel switches, for different families of \ac{BI} sets.}
\label{tab:strategy_overview}
\end{table}

\chapter{Evaluation}
\label{sec:numerical_eval}

In the following we present a numerical evaluation of the proposed discovery strategies. Complementing optimality results w.r.t. \ac{EMDT} and makespan for \ac{BI} sets from $\mathbb{F}_3$, presented in Chapter~\ref{sec:strategies}, we evaluate the proposed algorithms over \ac{BI} sets for which \ac{EMDT}-optimality and/or makespan-optimality is not guaranteed, also considering performance metrics such as number of channel switches, and energy consumption. We will show that the performance of \ALG{} algorithms is very close to the optimum in most settings. Given their low complexity, they thus represent an attractive solution to the problem of neighbor discovery for a broad range of deployment scenarios.

We will describe the evaluation setting in Section~\ref{sec:selection_parameters}, the performance metrics in Section~\ref{sec:performance_metrics}, and the results of the evaluation in Section~\ref{sec:ev_results}.


\section{Evaluation Setting}
\label{sec:selection_parameters}

In the following, we evaluate and compare \ALG{} discovery algorithms with PSV, CHAN TRAIN, and GENOPT strategies. In order to evaluate the proposed algorithms over \ac{BI} sets from $\mathbb{F}_1$ and $\mathbb{F}_2$ we draw random samples from these families of \ac{BI} sets until the confidence intervals for the studied performance metrics are sufficiently small. Note, however, that in order to include GENOPT, which is based on solving an \ac{ILP}, into the evaluation, we have to restrict the size of the studied scenarios, to keep the computational effort feasible. The number of decision variables of GENOPT depends mainly on the $LCM(B)$ and the number of channels $|C|$. Consequently, we have to restrict the maximum number of channels, as well as the number of elements in studied \ac{BI} sets, and their size. In particular, we vary the number of channels between 2 and 12. The minimum size of \ac{BI} sets is set to $\lvert B\rvert=3$, since the special case $\lvert B\rvert=2$ can be solved in an \ac{EMDT}-optimal and makespan-optimal way for arbitrary \ac{BI} sets from $\mathbb{F}_1$, as show in Section~\ref{sec:opt2}. Please note that due to the fact that a listening schedule for a \ac{BI} set $B=\left\{b_1,\ldots,b_n\right\}$ can be transformed into a schedule for a \ac{BI} set $B'=\left\{c\cdot b_1,\ldots,c\cdot b_n\right\}$, substituting $\tau'=c\cdot\tau$, we only consider \ac{BI} sets with $GCD(B)=1$. Still, the results of the evaluation are valid for a much broader range of \ac{BI} sets, including all sets obtained by multiplying individual \acp{BI} with a constant factor.

In the following we describe the procedure used to sample families of \ac{BI} sets $\mathbb{F}_1$ and $\mathbb{F}_2$ as well as the approach to computing \ac{ILP}-based listening schedule for the GENOPT strategy.

\subsection{Sampling $\mathbb{F}_1$}
\label{subsec:GenerationF1}

We draw random samples $B\in\mathbb{F}_1$ as follows. We first draw $\lvert B\rvert$ from a uniform distribution over $\left[3,6\right]$. We then draw individual \acp{BI} from a uniform distribution over $\left[1,10\right]$. The selected \acp{BI} are then divided by their GCD. If the resulting \ac{BI} set is new, it is used for evaluation. In total, this approach results in 775 unique \ac{BI} sets.

\subsection{Sampling $\mathbb{F}_2$}
\label{subsec:GenerationF2}

We draw random samples $B\in\mathbb{F}_2$ as follows. For each number from $\left[1,256\right]$ we first compute the power set of its factors. From the computed power sets, we select subsets whose cardinality is between 3 and 8, that contain the number itself, and whose GCD is one. In total, we obtain 259286 unique \ac{BI} sets.

\subsection{Computing GENOPT}
\label{sec:computing_genopt}

The computational complexity of GENOPT strongly depends on the maximum number of time slots allowed for the solution. Corollary~\ref{cor:mdtcor1} states that the makespan of \ac{EMDT}-optimal schedules is bounded from above by $LCM(B)\cdot\lvert C\rvert$. Note that $LCM(B)$ is $\mathcal{O}\left(\max(B)!\right)$ and thefore computing solutions over such a higher number of time slots can be extremely time consuming. However, we observed that in the vast majority of cases, optimal solutions can be found within $2\cdot\max(B)$ time slots. Therefore, in order to reduce solving time, we execute GENOPT iteratively, increasing the upper bound on the maximum time slots. More precisely, we initialize the Gurobi solver~\cite{Gurobi} with the best solution found by executing the \GreedyRnd{} strategy multiple times (1000 times for $B\in\mathbb{F}_1$, 50 times for $B\in\mathbb{F}_2$, because of higher complexity due to larger considered $\max(B)$) and restrict the maximum number of time slots to the number of time slots required by this solution. In the second step, the number of slots is doubled, while in the final third step, it is set to $LCM(B)\cdot\lvert C\rvert$. For the first and second iteration the time limit for the optimization is set to 1800 seconds, while for the final iteration it is set to 7200 seconds. If, however, after the time limit is exceeded the optimality gap is still greater than 3\% for $B \in \mathbb{F}_1$ and 1\% for $B \in \mathbb{F}_2$, the optimization is resumed until the optimality gap is sufficiently small.

\section{Performance Metrics}
\label{sec:performance_metrics}

The discovery strategies have been compared w.r.t. the performance metrics described in the following. 


\paragraph{\ac{EMDT}}

The \ac{EMDT} is the expected mean detection time of all neighbor networks, given that all channels, \acp{BI}, and offsets have an equal probability to be selected by a network (see Chapter~\ref{sec:preliminaries} for more details). Here expected value is over all possible network environments for a given set of channels and set of \acp{BI}, while mean is over the neighbor networks in a particular environment (see Section~\ref{sec:emdt} for more details). In the following, \ac{EMDT} results are depicted normalized by their optimal value, or, to be more precise, by the objective boundary obtained while generating the listening schedule for GENOPT (which is within 1\% to 3\% from the optimum, see Section~\ref{sec:computing_genopt}).

The \ac{EMDT} is an important metric in scenarios in which a device has to discover a subset of its neighbors, e.g. in \acp{DTN} in which devices have to detect suitable forwarders.

\paragraph{Number of channel switches}

This metric refers to the number of times a device has to change the listening channel when executing a complete listening schedule. In the following, it is depicted normalized by the minimum number of switches for a complete schedule, which is $(|C| - 1)$. When devices perform a channel switch they are in a deaf period in which they are not able to receive any frames, which might lead to losing beacons. A low number of channel switches is especially important if the ratio of the channel switching time to the duration of the time slot $\tau$ is large. In~\cite{Karowski11, Karowski13} it has been shown that listening schedules with higher number of channel switches result in higher discovery times when being executed in a realistic environment on IEEE~802.15.4 devices, even though under ideal conditions, when the duration of a channel switch is assumed to be 0, they exhibit identical performance. 

\paragraph{Makespan}

The makespan is the total number of time slots required by a listening schedule to detect all networks that are using any configuration $\kappa\in K_{BC}$ (see Definition~\ref{def:makespan}). In other words, it is the number of time slots required by the schedule to detect all neighbors in a complete environment. Note that in environments that are not complete, makespan does not equal the detection time of the last neighbor. Still, the device performing the discovery cannot know when all neighbors are detected and, therefore, has to execute the schedule until all potential neighbors using a configuration $\kappa\in K_{BC}$ are detected. In the following, makespan is depicted normalized to the optimum makespan $\max(B)\cdot\lvert C\rvert$, established by Corollary~\ref{cor:maxBC}.

\paragraph{Number of active time slots}

The number of active slots of a schedule is the number of time slots during which a device is actually listening on any channel. It is proportional to the amount of energy required to execute a schedule. The number of active time slots is always less or equal than the makespan due to the fact that a listening schedule might contain idle time slots during which no scan is allocated on any channel. In particular, makespan-optimal schedules are also optimal w.r.t. the number of active time slots. In the following, the number of active time slots is depicted normalized to its optimum value $\max(B)\cdot\lvert C\rvert$.

\section{Evaluation Results}
\label{sec:ev_results}

In the following, we extend the analytical optimality results established in Chapter~\ref{sec:strategies} by numerical evaluations over sets of \acp{BI} for which optimality cannot be proven. Moreover, in addition to \ac{EMDT} and makespan considered in Chapter~\ref{sec:strategies}, we extend the set of considered performance metrics by the number of channel switches and number of active time slots. As described in more details in Section~\ref{sec:selection_parameters}, we randomly sample the families of \acp{BI} $\mathbb{F}_1$ and $\mathbb{F}_2$, and vary the number of channels $\lvert C\rvert$ between 2 and 12. For each $\lvert C\rvert\in[2,12]$ we repeat the evaluation for 150 different \ac{BI} sets. The values depicted in the following are mean values computed over 150 iterations with different \ac{BI} sets, accompanied by the confidence intervals for the confidence level 95\%.

\subsection{Family of \ac{BI} Sets $\mathbb{F}_2$}
\label{subsec:results_f2}

In Chapter~\ref{sec:strategies} we showed that \ALG{} algorithms are makespan-optimal for \ac{BI} sets from $\mathbb{F}_2$. Moreover, the results of the numerical evaluation revealed that CHAN TRAIN is makespan-optimal over the evaluated \ac{BI} sets from $\mathbb{F}_2$. Therefore, in the following, we only show results for the \ac{EMDT} and the number of channel switches. (Note that the number of active time slots is also optimal for makespan-optimal schedules.)

Figure~\ref{fig:F2_discTime} depicts the normalized \ac{EMDT} of the evaluated discovery strategies. We observe that \ALG{} algorithms are within 2\% of the optimum and approximate the optimum even further when the number of channels increases. CHAN TRAIN has a slightly higher \ac{EMDT}, which is, however, still within 3\% of the optimum. In contrast, PSV has a significantly larger \ac{EMDT}, reaching 400\% of the optimum, and diverging when the number of channels increases. Note that the 4 studied \ALG{} algorithms are almost indistinguishable w.r.t. to their \ac{EMDT} at the given confidence level.

Figure~\ref{fig:f2_numeric_results} shows the number of channel switches normalized to the minimum value $\lvert C\rvert -1$. We can observe that the design of CHAN TRAIN aiming at reducing the number channel switches is successful at achieving its goal. It results in the second lowest number of channel switches with about 30 times the number of switches as compared to optimum achieved by PSV. Furthermore the distance to the optimum number is constant with increasing number of channels. Further we observe that \ALG{} algorithms \GreedyTrainRnd{} and \GreedyTrainDeter{}, that prioritize the channel allocated on the previous time slot if possible, allow to reduce the number of channel switches w.r.t. their versions that do not do the prioritization. We also observe that GENOPT and \GreedyRnd{} have relatively poor performance, requiring up to a factor of 200 more switches than the optimum value achieved by PSV.

\begin{figure*}
	\centering
		\subfloat[\ac{EMDT}]{
        \includegraphics[width=\evalFigWidth\textwidth]{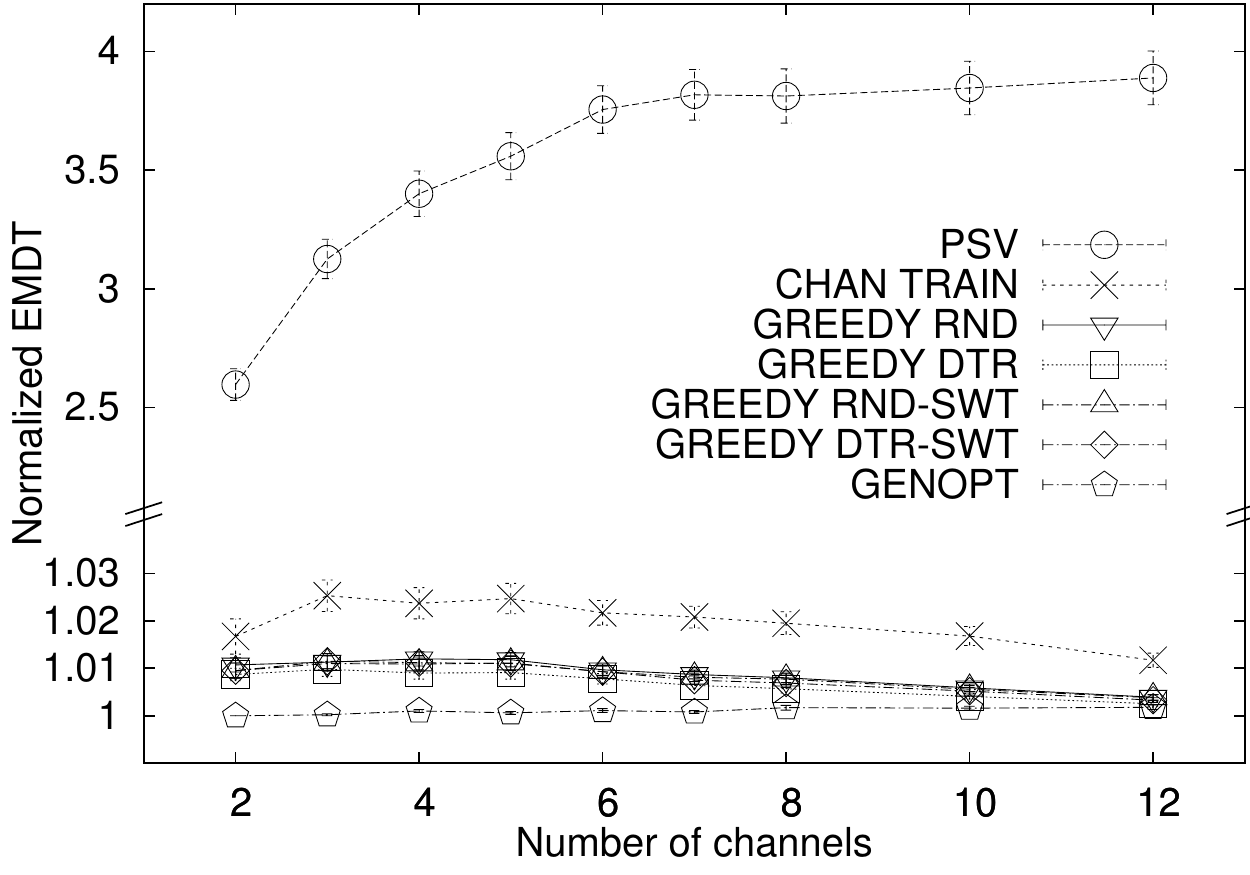}
				\label{fig:F2_discTime}
    }
		\hfill
		\subfloat[Number of Channel Switches]{	
        \includegraphics[width=\evalFigWidth\textwidth]{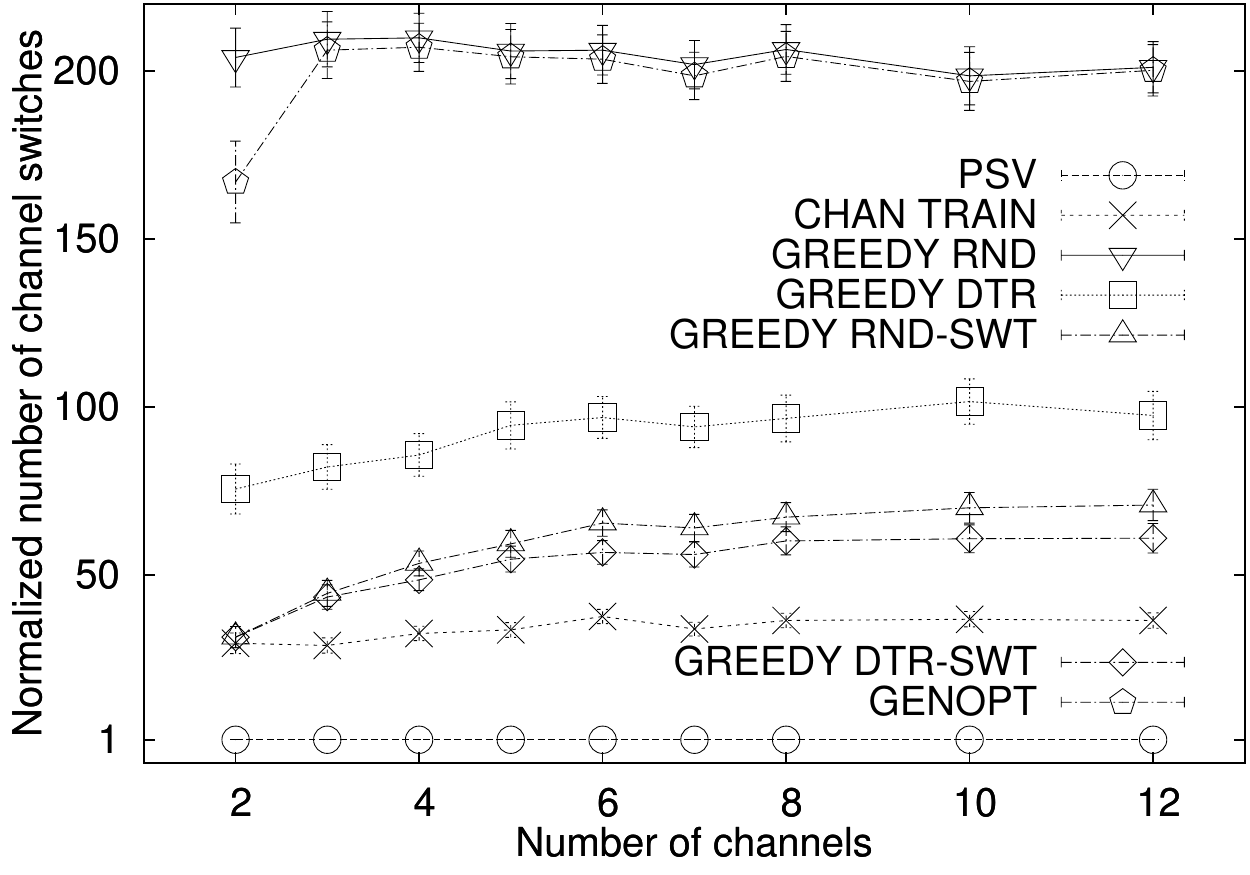}
				\label{fig:F2_chanJumps}
		}
\caption{Evaluation results for the family of \ac{BI} sets $\mathbb{F}_2$ (see Section~\ref{subsec:results_f2} for details).}
 \label{fig:f2_numeric_results}
\end{figure*}

\subsection{Family of \ac{BI} Sets $\mathbb{F}_1$}
\label{subsec:results_f1}


Figure~\ref{fig:F1_discTime} shows the normalized \ac{EMDT} for the family of \ac{BI} sets $\mathbb{F}_1$. We observe that the performance of the individual discovery algorithms, realative to each other, did not significantly change. We also observe that their normalized \ac{EMDT} slightly improved. On the one hand, a potential explanation for this is the less regular structure of \ac{BI} sets in $\mathbb{F}_1$. As shown in Corollary~\ref{cor:mdtcor1}, the upper bound for the makespan of \ac{EMDT}-optimal shedules is $LCM(B)\cdot\lvert C\rvert$, while in $\mathbb{F}_2$ it is $\max(B)\cdot\lvert C\rvert$ (as shown in Corollary~\ref{cor:mdtcor2}). On the other hand, the evaluation for \ac{BI} sets from $\mathbb{F}_1$ was performed with considerably smaller \acp{BI}, in order to make the evaluation of GENOPT computationally feasible. Still, while \ac{EMDT} of \ALG{} strategies are within 1\% of the optimum, \ac{EMDT} of PSV reaches 160\% of the optimum is is increasing with the number of channels.

Figure~\ref{fig:F1_chanJumps} displays the number of channel switches. Similar to \ac{EMDT}, we observe that the relative performance of the discovery strategies remained similar as for $\mathbb{F}_2$. Also, the distance to the optimum has decreased for all approaches, which, again, might be due to the smaller \acp{BI} used for evaluation.


\begin{figure*}
	\centering
		\subfloat[\ac{EMDT}]{
        \includegraphics[width=\evalFigWidth\textwidth]{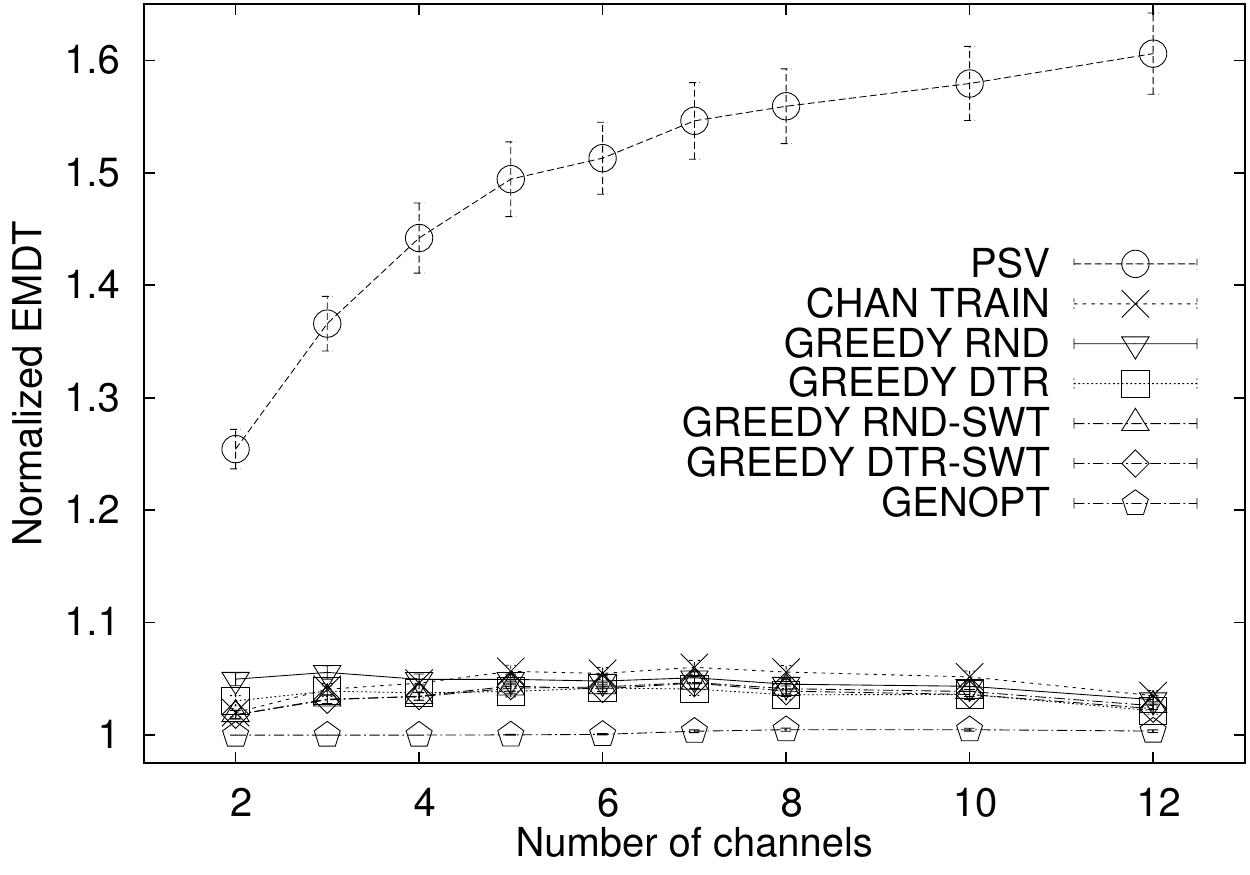}
				\label{fig:F1_discTime}
    } 
	\hfill
		\subfloat[Number of Channel Switches]{		
        \includegraphics[width=\evalFigWidth\textwidth]{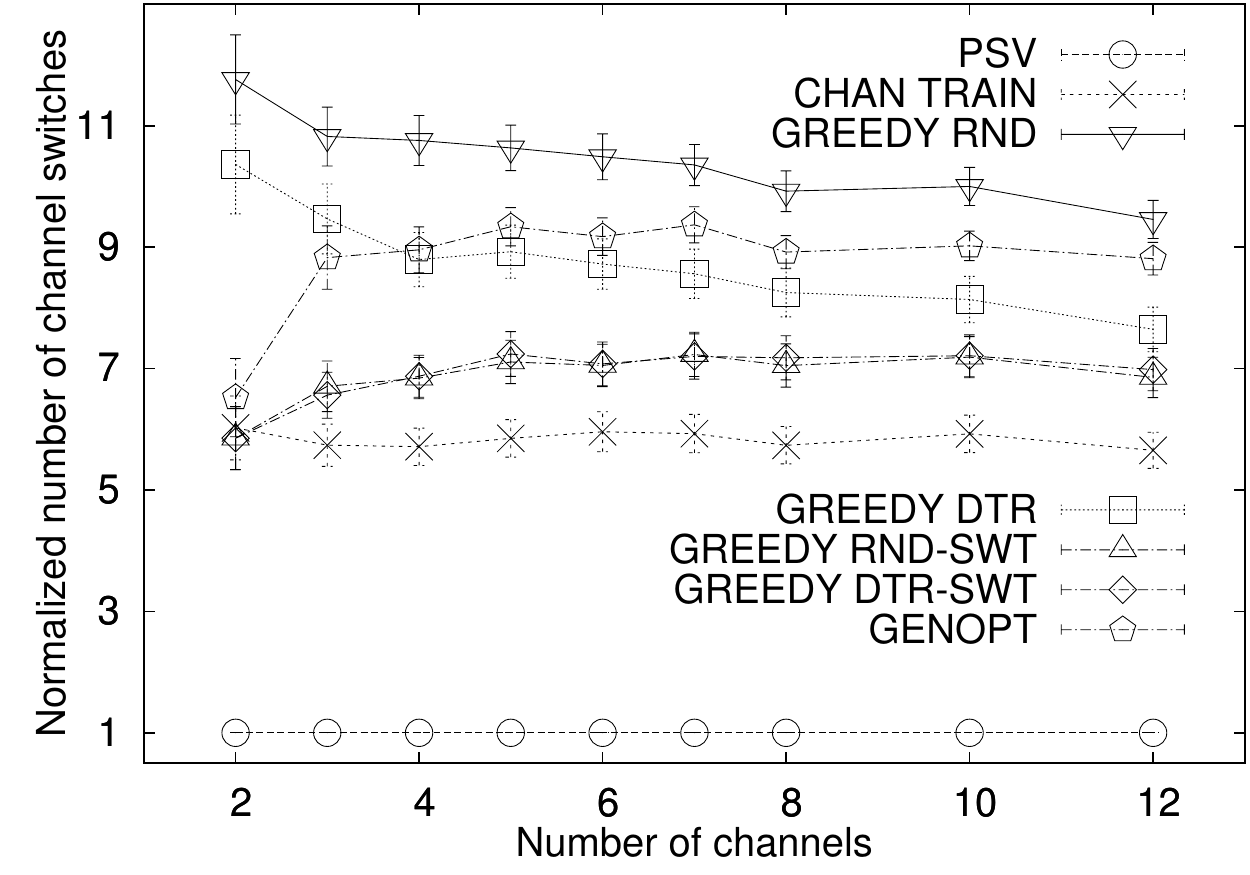}
				\label{fig:F1_chanJumps}
		}
\caption{Evaluation results for the family of \ac{BI} sets $\mathbb{F}_1$ (see Section~\ref{subsec:results_f1} for details).}
 \label{fig:f1_numeric_results1}
\end{figure*}

Figure~\ref{fig:F1_makeSpan} shows the normalized makespan for the family of \ac{BI} sets $\mathbb{F}_1$. For all strategies the normalized makespan improves with the increasing number of channels. E.g., for GENOPT the gap reduces from about 15\% for two channels to about 1\% for 12 channels. The listening schedules of the considered \ALG{} approaches and CHAN TRAIN result in similar makespan at the considered confidence level.

Since makespan is only providing a measure for the last discovery time but not the actual energy usage, Figure~\ref{fig:F1_listenSlots} depicts the normalized number of active slots. The results for all strategies are very similar to those for the makespan, except that they are shifted by about 5\% meaning that the schedules of the \ALG{} approaches, CHAN TRAIN and GENOPT consist of about 5\% idle slots in which no scan is scheduled on any channel due to the fact that no new configurations will be discovered.

\begin{figure*}
	\centering
		\subfloat[Makespan]{
       \includegraphics[width=\evalFigWidth\textwidth]{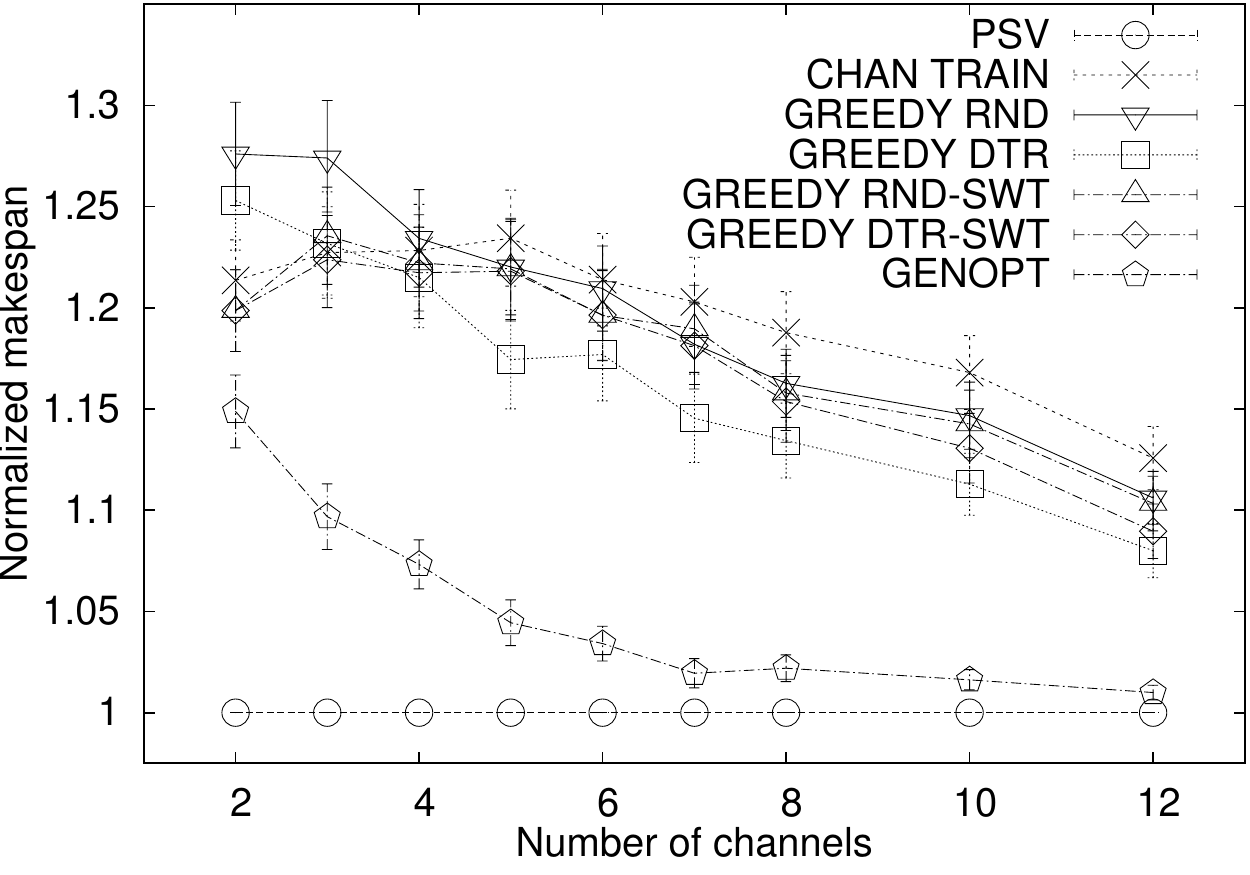}
				\label{fig:F1_makeSpan}
    }
			\hfill
		\subfloat[Number of Active Slots]{
        \includegraphics[width=\evalFigWidth\textwidth]{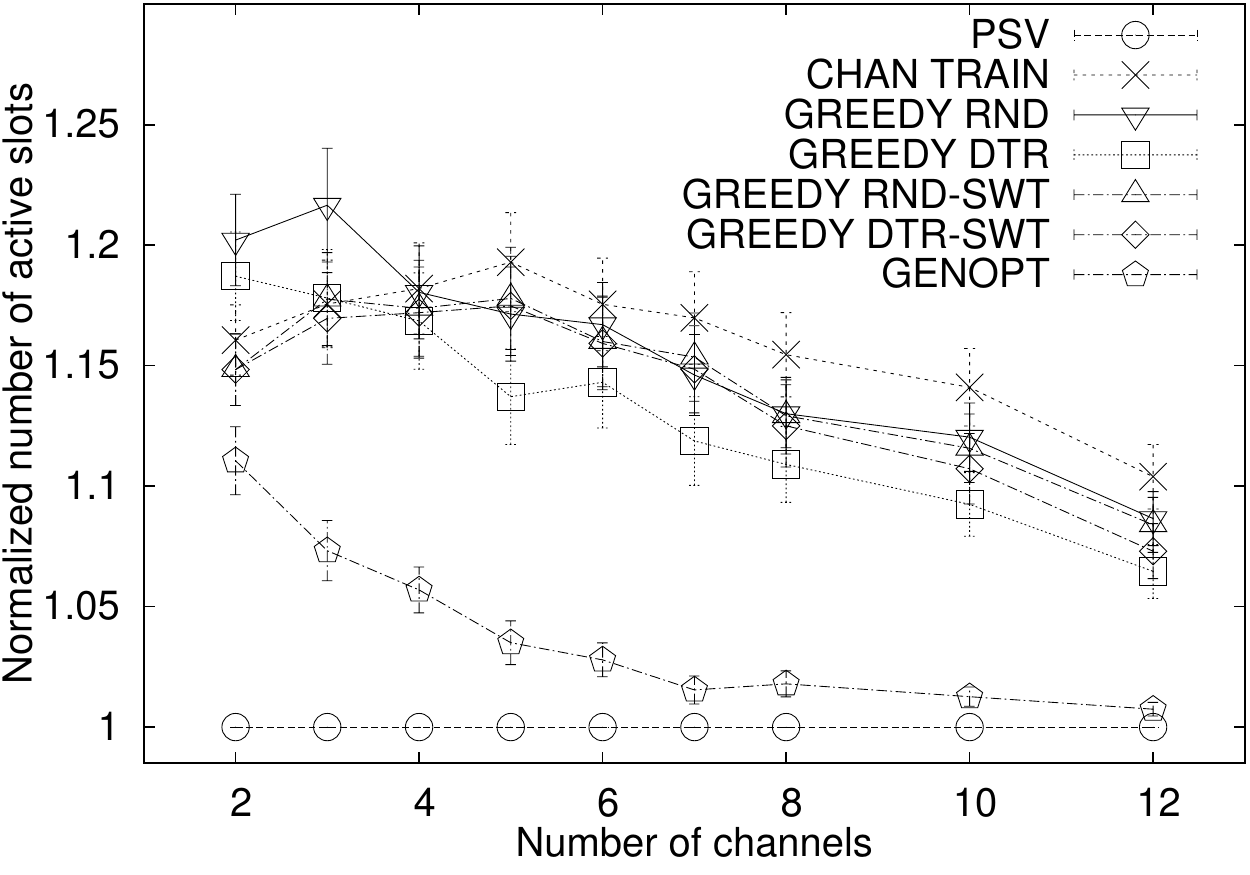}
				\label{fig:F1_listenSlots}
    }	
\caption{Evaluation results for the family of \ac{BI} sets $\mathbb{F}_1$ (see Section~\ref{subsec:results_f1} for details).}
 \label{fig:f1_numeric_results2}
\end{figure*}

\chapter{Selection of Beacon Intervals to Support Efficient Neighbor Discovery}
\label{sec:biSelection}

In the preceding sections we presented and evaluated discovery strategies that support arbitrary finite \ac{BI} sets from $\mathbb{F}_1$. However, the low-complexity GREEDY approaches are \ac{EMDT}-optimal only for \ac{BI} sets from $\mathbb{F}_3$ and makespan-optimal for \ac{BI} sets from  $\mathbb{F}_2$. Although we also presented an \ac{ILP}-based approach GENOPT which is \ac{EMDT}-optimal for \ac{BI} sets from $\mathbb{F}_1$, it suffers from high computational complexity and can only be applied offline and to network environments limited in size.

In this section we discuss and provide intuition for problems arising when constructing listening schedules for \ac{BI} sets in $\mathbb{F}_1$ and $\mathbb{F}_2$, and provide a recommendation for the selection of \ac{BI} sets that support efficient neighbor discovery. This recommendation may be useful for the development of new technologies and communication protocols for wireless communication that use periodic beacon frames for management or synchronization purposes or in case of deploying devices using existing technologies that support a wide range of \acp{BI}, such as IEEE~802.11. In particular, we identified two following issues that can arise when constructing listening schedules.

\begin{itemize}
	\item The makespan of the schedule is greater than $\max(B)\cdot\lvert C\rvert$ time slots (i.e., it is not makespan-optimal).
	\item The \acp{BI} are not discovered in ascending order (i.e., the schedule is not recursive).
\end{itemize}
 
The first problem of listening schedules being not makespan-optimal is caused by idle slots and redundant scans. During idle slots no scan is scheduled due to the fact that no new configurations can be discovered on any channel. In contrast, during redundant slots, a part of neighbors that send beacons during this time slot has already been discovered. Redundant scans are caused by the fact that information acquired in previously scanned slots cannot be completely reused for the discovery of other configurations on the same channel. 

Figure~\ref{fig:GeneralBISetProblem_AddSlots} shows an optimized listening schedule generated by GENOPT for $B=\{1,2,3\}$ and $|C| = 2$ and is an example for a schedule containing redundant scans. Even though $\max(B) = 3$ there are four time slots scanned on channel $c_1$. The reason for the redundant scans is a conflict between the attempt to decrease the discovery time for smaller \acp{BI} in order to minimize \ac{EMDT} and the missing possibility of reusing the information acquired by scanning the time slot scheduled for smaller \acp{BI} $b_1$ and $b_2$ for the discovery of neighbors operating with the larger \ac{BI} $b_3$. From time slot $t_1$ to $t_4$, the schedule finishes the discovery of neighbors using $b_1$ and $b_2$ with a minimum \ac{EMDT}. However, the scan at $t_4$ will not contribute to the discovery of neighbors operating with $b_3$ on $c_1$ since $t_1$ was already scanned on $c_1$. Another listening schedule achieving equal \ac{EMDT} is described by $\mathcal{L} = \{(1,1),(0,3),(1,2)\}$ (see Section~\ref{sec:system} for definitions and notation). In total it requires one time slot less but also increases the discovery time for neighbors using $b_2$ on $c_1$ while decreasing the discovery time for $b_3$ on $c_0$. This schedule suffers from the second identified issue that will be discussed further below.

%


Figure~\ref{fig:GeneralBISetProblem_Idle} depicts an example for a listening schedule generated by GENOPT for $B=\{2,3,4\}$ and $|C| = 2$ containing an idle slot. The schedule first completes the discovery of neighbors with $b_2$ on $c_1$, then $b_2$ and $b_3$ on $c_0$, and then $b_3$ on $c_1$ during time slot $t_6$. Due to the fact that scanning $t_6$ on $c_1$ does not result in any new information regarding neighbors operating with $b_4$, since $t_2$ has already been scanned on $c_1$, time slots $t_7$ and $t_8$ have to be scanned to finish the discovery of all neighbors operating on $c_1$. Time slot $t_9$ is not assigned to channel $c_0$ because $t_5$ has been already scanned and therefore the time slot will be idle.

Similar to the previous example illustrating redundant scans, there also exists an alternative listening schedule $\mathcal{L} = \{(1,2),(0,4),(1,3)\}$ for the given parameter set resulting in the same \ac{EMDT}. This schedule does not contain any idle time slots but, again, completes the discovery of neighbors with larger \acp{BI} before detecting all neighbors with lower \acp{BI}. In this example the discovery time of neighbors operating with $b_3$ on $c_1$ is increased while the discovery of neighbors using $b_4$ on $c_1$ is speeded up.


The second identified issue relates to the order in which configurations with specific \acp{BI} are discovered. In order to reduce \ac{EMDT} for a given \ac{BI} set $B$ and channel set $C$ a listening schedule shall discover neighbors in ascending order of their \acp{BI}. We call such schedules recursive (see Definition~\ref{def:recursive_schedule}).

Examples of this issue were already mentioned in the description of the first issue regarding alternative schedules to the solutions shown in Figures~\ref{fig:GeneralBISetProblem_AddSlots} and \ref{fig:GeneralBISetProblem_Idle}. Another example is shown in Figure~\ref{fig:GeneralBISetProblem_Order}. The discovery of neighbors using $b_1$ or $b_2$ is completed on channel $c_1$ even before a time slot has been scanned on channel $c_2$.
Due to the composition of \acp{BI} in set $B$ a recursive listening schedule does not exist.

\begin{figure*}[ht]
\centering
		\subfloat[Redundant scans \newline \hspace{\textwidth} $B=\{1,2,3\} \quad |C| = 2$]{
        \includegraphics[scale=0.9]{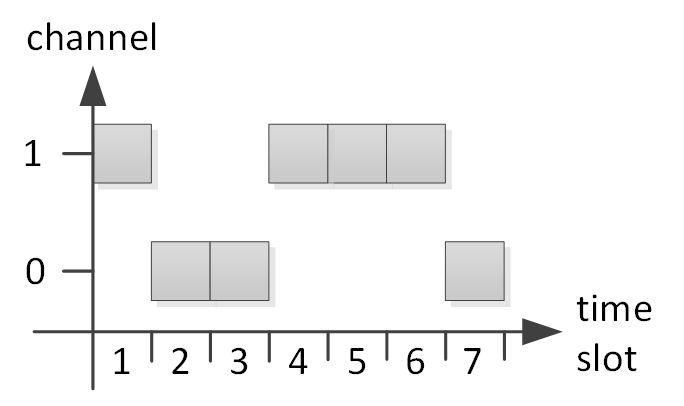}
				\label{fig:GeneralBISetProblem_AddSlots}
    }
		\subfloat[Idle slots \newline \hspace{\textwidth} $B=\{2,3,4\} \quad |C|=2$ ]{
        \includegraphics[scale=0.9]{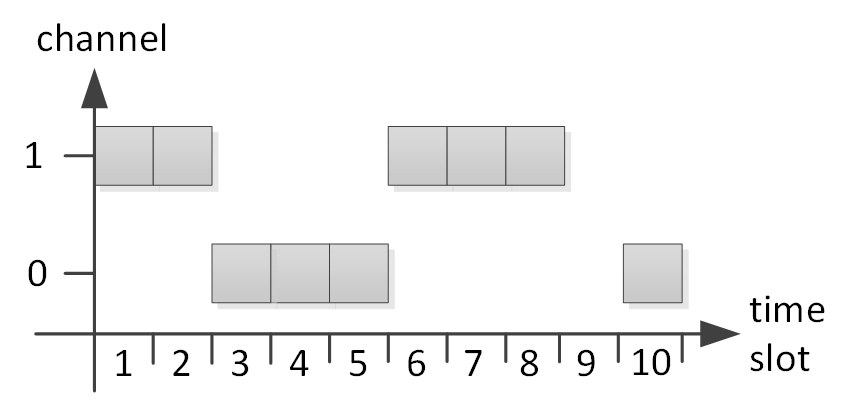}
				\label{fig:GeneralBISetProblem_Idle}
    } \hfill
		\subfloat[Order of \ac{BI} discovery \newline \hspace{\textwidth} $B=\{1,2,3\} \quad |C| = 3$]{
       \includegraphics[scale=0.9]{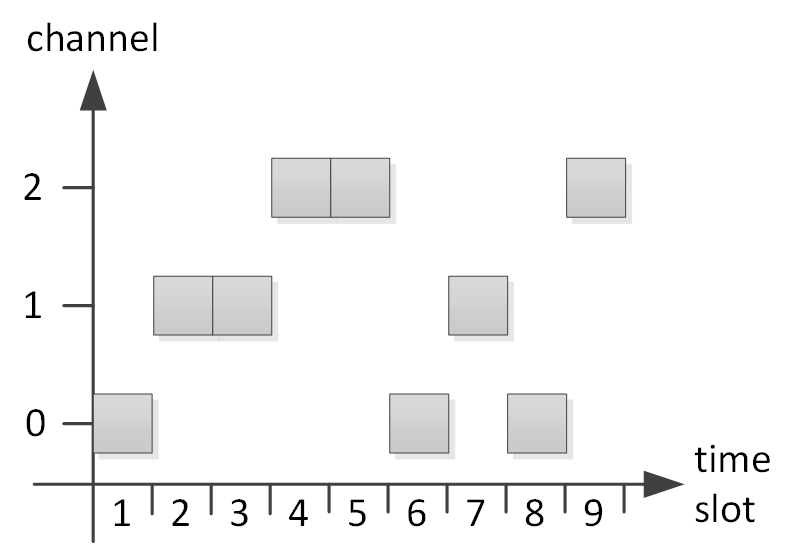}
				\label{fig:GeneralBISetProblem_Order}
    }
\caption{Issues arising when using arbitrary \ac{BI} sets.}
 \label{fig:GeneralBISetProblem}
\end{figure*}


In the following we will give a recommendation on the selection of \acp{BI} sets that help avoiding the issues described above. To support efficient neighbor discovery by eliminating idle time slots as well as redundant scans, and to enable generation of recursive listening schedules, we recommend to prefer \ac{BI} sets from $\mathbb{F}_3$ over those from $\mathbb{F}_1$ and $\mathbb{F}_2$. Note that for the family of \ac{BI} sets $\mathbb{F}_3$ our low-complexity \ALG{} algorithms generate recursive listening schedules that result in optimal \ac{EMDT} and optimal makespan.

Nevertheless, if a higher diversity of \acp{BI} is required, we recommend to select $\mathbb{F}_2$. For this family \ALG{} algorithms and CHAN TRAIN generate listening schedules with optimal makespan and a slightly higher \ac{EMDT} than the optimal value. Depending on whether \ac{EMDT} or the number of channel switches shall be prioritized, we recommend to use either one of the \ALG{} approaches or, alternatively, the CHAN TRAIN strategy.

In case a scenario requires using a \ac{BI} set from $\mathbb{F}_1$, we suggest to use either a \ALG{} algorithm or the CHAN TRAIN strategy. Both options allow to generate listening schedules with an \ac{EMDT} that is within 5\% of the optimal value for the evaluated \ac{BI} sets. If a dynamic computation of listening schedules directly on the devices themselves is not required and the \ac{BI} set $B$ and channel set $C$ are static, GENOPT can be used to compute a schedule offline and store it on the devices.


\chapter{Conclusion}
\label{sec:conclusion}

This paper extends our previous work on asynchronous passive multi-channel discovery. It presents novel discovery strategies that are not limited to \acp{BI} defined in the IEEE~802.15.4 standard but support a broad range of \ac{BI} sets. In particular, we characterizes a family of low-complexity algorithms, named \ALG{}, minimizing makespan and \ac{EMDT} for the broad family of \ac{BI} sets $\mathbb{F}_3$, where each \ac{BI} is an integer multiple of all smaller \acp{BI}. Notably, this family completely includes and significantly extends \ac{BI} sets supported by the IEEE~802.15.4 standard. In addition, we develop an \ac{ILP}-based approach minimizing \ac{EMDT} for arbitrary \ac{BI} sets that, however, exhibits high computational complexity and memory consumption and is, therefore, only applicable for offline computation and for network environments with limited size.

In addition to analytically proving optimality for the developed approaches for \ac{BI} sets from $\mathbb{F}_3$, we performed an extensive numerical evaluation on two more general families of \ac{BI} including the most general ones. We compared our strategies with the passive discovery strategy defined by the IEEE~802.15.4 standard and showed that our approaches exhibit superior performance w.r.t. several performance metrics. Finally, we provide recommendations on the selection of the \ac{BI} sets that are as nonrestrictive as possible, at the same time allowing for an efficient neighbor discovery.

Our future work will focus on collaboration between devices by sharing gossip information about discovered neighbors in their beacon messages. By including this information into listening schedules computation, devices will be able to speed up the discovery of their neighbors. In addition, we will evaluate the developed approaches in realistic environments using simulations and experiments.

\acrodef{BI}{Beacon Interval} \acrodefplural{BI}[BI's]{Beacon Intervals}
\acrodef{BO}{Beacon Order} \acrodefplural{BO}[BO's]{Beacon Orders}
\acrodef{DTN}{Delay Tolerant Network} \acrodefplural{DTN}[DTN's]{Delay Tolerant Networks}
\acrodef{GCD}{Greatest Common Divisor} 
\acrodef{EMDT}{Expected Mean Discovery Time} \acrodefplural{EMDT}[EMDT's]{Expected Mean Discovery Times}
\acrodef{ILP}{Integer Linear Program} \acrodefplural{ILP}[ILP's]{Integer Linear Programs}
\acrodef{LCM}{Least Common Multiple} \acrodefplural{LCM}[LCM's]{Least Common Multiples}
\acrodef{LP}{Linear Program} \acrodefplural{LP}[LP's]{Linear Programs}
\acrodef{MAC}{Media Access Control}
\acrodef{MDT}{Mean Discovery Time} \acrodefplural{MDT}[MDT's]{Mean Discovery Times}
\acrodef{TU}{Time Unit} \acrodefplural{TU}[TU's]{Time Units}

\bibliography{references}

\newpage
\begin{appendix}
\end{appendix}

\end{document}